\documentclass[10pt,onecolumn]{IEEEtran}

\usepackage{amsmath,amssymb,amsfonts,graphicx}
\usepackage{latexsym}
\usepackage{amsthm}
\usepackage{url}

\usepackage{algorithm}
\usepackage{algpseudocode}

\usepackage{caption}
\usepackage{subcaption}

\newcommand{\database}{{\mathcal{D}^n}}    
\newcommand{\noise}{{X}}    
\newcommand{\KM}{{\mathcal{K}}}    
\newcommand{\e}{{\epsilon}}    
\newcommand{\loss}{{\mathcal{L}}}    
\newcommand{\D}{{\Delta}}    
\newcommand{\p}{{\mathcal{P}}}  

\newcommand{\mt}{{r}}  

\newcommand{\nm}{{\p}}  

\newcommand{\sP}{{\mathcal{SP}}}  
\newcommand{\psym}{{\p_{\text{sym} }}}
\newcommand{\sPsymi}{{\sP_{i, \text{sym} }}}
\newcommand{\gammap}{{\gamma'}}

\newcommand{\R}{{\mathbb{R}}}  
\newcommand{\N}{{\mathbb{N}}}  
\newcommand{\Z}{{\mathbb{Z}}}  

\newcommand{\pa}{{ \sP_{i, \text{md}}  }}  
\newcommand{\pb}{{ \sP_{i, \text{pd}}  }}  
\newcommand{\pc}{{ \sP_{i, \text{fr}}  }}  
\newcommand{\pd}{{ \sP_{i, \text{step}} }}  

\newcommand{\pe}{{ \sP_{\text{mono}} }}  
\newcommand{\pdissym}{{ \sP_{\text{sym}} }}  
\newcommand{\pf}{{ \sP_{\text{pd}} }}  
\newcommand{\pj}{{ \sP_{\text{step}} }}  

\newcommand{\gammaheu}{{ \tilde{\gamma} }}  

\DeclareMathOperator*{\argmax}{arg\,max}

\newtheorem{theorem}{Theorem}
\newtheorem{lemma}[theorem]{Lemma}
\newtheorem{definition}{Definition}
\newtheorem{corollary}[theorem]{Corollary}
\newtheorem{property}{Property}

\begin{document}

\title{The Optimal Mechanism in Differential Privacy}

\author{
\authorblockN{Quan Geng, and Pramod Viswanath}\\
\authorblockA{Coordinated Science Laboratory and Dept. of ECE \\
University of Illinois, Urbana-Champaign, IL 61801 \\
Email: \{geng5, pramodv\}@illinois.edu} }

\maketitle

\begin{abstract}

Differential privacy is a  framework to quantify to what extent individual privacy in a statistical database is preserved while releasing useful aggregate information about the database. In this work we study the fundamental tradeoff between privacy and utility in differential privacy. 
We derive the optimal $\epsilon$-differentially private mechanism for  single real-valued query function under a very general utility-maximization (or cost-minimization) framework. The class of  noise probability distributions in the optimal mechanism has  {\em staircase-shaped} probability density functions which are symmetric (around the origin), monotonically decreasing and geometrically decaying. The staircase mechanism can be viewed as a {\em geometric mixture of uniform probability distributions}, providing a simple algorithmic description for the mechanism.  Furthermore, the staircase mechanism naturally generalizes to discrete query output settings as well as more abstract settings. We explicitly derive the parameter of the optimal staircase mechanism for $\ell_1$ and $\ell_2$ cost functions.   Comparing the optimal performances with those of the usual Laplacian mechanism, we show that in the high privacy regime ($\epsilon$ is small),  the Laplacian mechanism is asymptotically optimal as $\epsilon \to 0$;  in the low privacy regime ($\epsilon$ is large), the minimum magnitude and second moment of noise  are $\Theta(\Delta e^{-\frac{\epsilon}{2}})$ and $\Theta(\Delta^2 e^{-\frac{2\epsilon}{3}})$ as $\epsilon \to +\infty$, respectively, while  the corresponding figures when using the Laplacian mechanism are $\frac{\Delta}{\epsilon}$ and $\frac{2\Delta^2}{\epsilon^2}$, where $\Delta$ is the sensitivity of the query function. We conclude that the gains of the staircase mechanism are more pronounced in the moderate-low privacy regime.


\end{abstract}

\section{Introduction} \label{sec:intro}

Differential privacy is a  formal framework to quantify to what extent individual privacy in a statistical database is preserved while releasing useful aggregate information about the database.
It provides strong privacy guarantees by requiring the indistinguishability of whether an individual is in the dataset or not based on the released information. The key idea of differential privacy is that the presence or absence
of any individual data in the database should not affect the final released statistical information significantly, and thus it can give strong privacy guarantees against an adversary with arbitrary auxiliary information. For motivation and background of differential privacy, we refer the readers to the  survey \cite{DPsurvey} by Dwork.

Since its introduction in \cite{DMNS06} by Dwork et.\ al., differential privacy has 
 spawned  a large body of research in  differentially private data-releasing mechanism design and performance analysis in various settings. Differential privacy is a privacy-preserving constraint imposed on the query output releasing mechanisms, and to make use of the released information, it is important to understand the fundamental tradeoff between utility(accuracy) and privacy.

In many existing works on studying the tradeoff between accuracy and privacy in differential privacy, the usual metric of accuracy is in terms of the variance, or magnitude expectation of the noise added to the query output.  For example, Hardt and Talwar \cite{geometry} study the tradeoff between privacy and error for answering a set of linear queries over a histogram in a differentially private way, where the error is defined as the worst  expectation of the $\ell^2$-norm of the noise among all possible query output.  \cite{geometry} derives lower and upper bounds on the error given the differential privacy constraint. Nikolov, Talwar and Zhang \cite{NTZ12} extend the result on the tradeoff between privacy and error to the case of  $(\e,\delta)$-differential privacy. Li et. al. \cite{Li10} study how to optimize linear counting queries under differential privacy, where the error is measured by the mean squared error of query output estimates, which corresponds to the variance of the noise added to the query output to preserve differential privacy.

More generally, the error can be a general function depending on the additive noise (distortion) to the query output.  Ghosh, Roughgarden,  and Sundararajan \cite{Ghosh09}  study a very general utility-maximization framework for a single count query with sensitivity one under differential privacy, where the utility (cost) function can be a general function depending on the noise added to the query output. \cite{Ghosh09} shows that there exists a universally optimal mechanism (adding geometric noise) to preserve differential privacy for a general class of utility functions under a Bayesian framework.  Brenner and Nissim \cite{Nissim10} show that for general query functions, no universally optimal differential privacy mechanisms exist. Gupte and Sundararajan  \cite{minimax10} generalize the result of \cite{Ghosh09}  to a minimax setting.

In this work, we study the fundamental tradeoff between utility and privacy under differential privacy, and derive the optimal differentially private mechanism for  general single real-valued query function, where the utility model is the same as the one adopted in \cite{Ghosh09} and \cite{minimax10}, and the real-valued query function can have arbitrary sensitivity.  Our results can be viewed as a generalization of  \cite{Ghosh09} and \cite{minimax10} to general real-valued query functions with arbitrary sensitivity. We discuss the relations of our work and the existing works in detail in Section \ref{subsec:connection}.

\subsection{Background on Differential Privacy}

The basic problem setting in differential privacy for statistical database is as follows: suppose a dataset curator is in charge of a statistical database which consists of records of many individuals, and an analyst sends a query request to the curator to get some aggregate information about the whole database. Without any privacy concerns, the database curator can simply apply the query function to the dataset, compute the query output, and send the result to the analyst. However, to protect the privacy of individual data in the dataset, the dataset curator should use a randomized query-answering mechanism such that the probability distribution of the query output does not differ too much whether any individual record is in the database or not.

Formally, consider a real-valued query function
\begin{align}
	q: \database \rightarrow \R,
\end{align}
where $\database$ is the set of all possible datasets. The real-valued query function $q$ will be applied to a dataset, and query output is a real number. Two datasets $D_1, D_2 \in \database$ are called neighboring datasets if they differ in at most one element, i.e.,  one is a proper subset of the other and the larger dataset contains just one additional element \cite{DPsurvey}. A randomized query-answering mechanism $\KM$ for the query function $q$ will randomly output a number with probability distribution depends on query output $q(D)$, where $D$ is the dataset.

\begin{definition}[$\e$-differential privacy \cite{DPsurvey}]
	A randomized mechanism $\KM$ gives $\e$-differential privacy if for all data sets $D_1$ and $D_2$ differing on at most one element, and all $S \subset \text{Range}(\KM)$,
	\begin{align}
	 	\text{Pr}[\KM(D_1) \in S] \le \exp(\e) \;  \text{Pr}[\KM(D_2) \in S], \label{eqn:dpgeneral}
	 \end{align}
	 where $\KM(D)$ is the random output of the  mechanism $\KM$ when the query function $q$ is applied to the dataset $D$.
\end{definition}

The differential privacy constraint \eqref{eqn:dpgeneral} essentially requires that for all neighboring datasets, the probability distributions of the output of the randomized mechanism should be approximately the same. Therefore, for any individual record, its presence or absence in the dataset will not significantly affect the output of the mechanism, which makes it hard for adversaries with arbitrary background knowledge to make inference on any individual from the released query output information. The parameter $\e \in (0, +\infty)$ quantifies how private the mechanism is: the smaller $\e$ is , the more private the randomized mechanism is.


\subsubsection{Operational Meaning of $\e$-Differential Privacy in the Context of Hypothesis Testing}

As shown by \cite{Zhou08}, one can interpret the   differential privacy constraint  \eqref{eqn:dpgeneral} in the context of hypothesis testing in terms of false alarm probability and missing detection probability. Indeed, consider a binary hypothesis testing problem over two neighboring datasets, $H_0: D_1 $ versus $H_1: D_2$, where an individual's record is in $D_2$ only. Given a decision rule, let $S$ be the decision region such that when the released output lies in $S$, $H_1$ will be rejected, and when the released output lies in $S^C$ (the complement of $S$), $H_0$ will be rejected. The false alarm probability $P_{FA}$ and the missing detection probability $P_{MD}$ can be written as
\begin{align}
	P_{FA} &= P(K(D_1) \in S^C), \\
	P_{MD} &= P(K(D_2) \in S).
\end{align}

Therefore, from \eqref{eqn:dpgeneral} we get
\begin{align}
	1 - P_{FA} \le e^{\e} P_{MD} .
\end{align}
Thus
\begin{align}
	 e^{\e} P_{MD} + P_{FA} \ge 1 .
\end{align}

Switch $D_1$ and $D_2$ in \eqref{eqn:dpgeneral}, and we get
	\begin{align}
	 	\text{Pr}[\KM(D_2) \in S] \le \exp(\e) \;  \text{Pr}[\KM(D_1) \in S] .
	 \end{align}

Therefore,
\begin{align}
	1 - P_{MD} \le e^{\e} P_{FA} ,
\end{align}
and thus
\begin{align}
	 P_{MD} + e^{\e} P_{FA} \ge 1 .
\end{align}

In conclusion,  we have
\begin{align}
  e^{\e} P_{MD} + P_{FA} &\ge 1 , \label{eq:famd1}\\
	P_{MD} + e^{\e} P_{FA} &\ge 1 \label{eq:famd2}.
\end{align}

The $\e$-differential privacy constraint implies that in the context of hypothesis testing,  $P_{FA}$ and $P_{MD}$ can not be both too small.




\subsubsection{Laplacian Mechanism}

The standard approach to preserving $\e$-differential privacy is to perturb the query output by adding random noise with Laplacian distribution proportional to the sensitivity $\Delta$ of the query function $q$, where the  sensitivity of a real-valued query function is defined as

\begin{definition}[Query Sensitivity \cite{DPsurvey}]
	For a real-valued query function $q: \database \rightarrow \R$, the sensitivity of $q$ is defined as
	\begin{align}
	\D := \max_{D_1,D_2 \in \database} | q(D_1) - q(D_2)|, \label{def:sensitivity}
\end{align}
for all $D_1,D_2$ differing in at most one element.
\end{definition}

Formally, the  Laplacian mechanism is:
\begin{definition}[Laplacian Mechanism \cite{DMNS06}]
	For a real-valued query function $q: \database \rightarrow \R$ with sensitivity $\Delta$, Laplacian mechanism will output
	\begin{align}
		K(D) := q(D) + \text{Lap}(\frac{\D}{\e}),
	\end{align}
 	where $\text{Lap}(\lambda)$ is a random variable with probability density function
 	\begin{align}
 		f(x) = \frac{1}{2\lambda} e^{-\frac{|x|}{\lambda}}, \quad \forall x \in \R.
 	\end{align}
\end{definition}


Consider two neighboring datasets $D_1$ and $D_2$ where $|q(D_1) - q(D_2)| = \D$. It is easy to compute the tradeoff between the false alarm probability $P_{FA}$ and the missing detection probability $P_{MD}$  under  Laplacian mechanism, which is
\begin{align}
P_{MD}  =
\begin{cases}
1 - e^{\e}P_{FA} & P_{FA} \in [0,   \frac{1}{2} e^{-\e}) \\
\frac{e^{-\e}}{4P_{FA}}  & P_{FA} \in [ \frac{1}{2} e^{-\e}, \frac{1}{2} ) \\
e^{-\e}(1 - P_{FA} ) & P_{FA} \in [\frac{1}{2}, 1]
\end{cases}\label{eqn:regionLap}
\end{align}




Since its introduction in \cite{DMNS06},  the Laplacian mechanism has become the standard tool in differential privacy and has been used as the basic building block in a number of works on differential privacy analysis in other more complex problem settings, e.g., \cite{HLM12, MM09, Xiao11, Huang12, McSherry10, Li10, Barak07, DKMMN06, DL09, Roth10, LO11, Smith2011, CM08, continual, Ding11, Hardt2010, Geo12, Ka11, Mironov12bit, Sarathy2011, Xiao2011ireduct, Dankar12, Friedman10, zhang2012functional, lei2011differentially, wasserman2010, dwork2010pan, Guptadp2010, blum2011fast, hsu2012distributed, hsu2012dp, blocki2012johnson, hardt2012beyond, hardt2012private, gupta2012iterative, kasi13, karwa2011private, cormode2012differentially}. Given this near-routine use of the query-output independent adding of Laplacian noise,  the following two questions are natural:
\begin{itemize}
	\item Is  query-output independent perturbation optimal?
	\item Assume query-output independent perturbation,  is Lapacian noise distribution optimal?
\end{itemize}

In this work we answer the above two questions. Our main result is that given an $\e$-differential privacy constraint, under a general utility-maximization (equivalently, cost-minimization) model:
\begin{itemize}
	\item adding query-output independent noise is indeed optimal (under a mild technical condition),
	\item the optimal noise distribution is {\em not} Laplacian distribution; instead, the optimal one has a {\em staircase-shaped} probability density function.
\end{itemize}

These results are derived under the following settings:
\begin{itemize}
\item when the domain of the query output is the entire real line or the set of integers;
\item nothing more about the query function is known beyond its global sensitivity;
\item either local sensitivity \cite{NRS07} of the query function is unknown or it is the same as global sensitivity (as in the important case of count queries).
\end{itemize}
If any of these conditions are violated (the output domain has sharp boundaries, or the local sensitivity deviates from the global sensitivity \cite{NRS07}, or we are restricted to specific query functions \cite{CM08}), then the optimal privacy mechanism need not be data or query output dependent.



\subsection{Problem Formulation}

We formulate a utility-maximization (cost-minimization) problem under the differential privacy constraint.

\subsubsection{Differential Privacy Constraint}
A general randomized releasing mechanism $\KM$ is a family of noise probability distributions indexed by  the query output (denoted by $t$), i.e.,
\begin{align}
 	\KM = \{\p_t : t \in \R \},
 \end{align}
and given dataset $D$, the mechanism $\KM$ will release the query output $t = q(D)$ corrupted by additive random noise with probability distribution $\p_t$:
\begin{align}
    \KM(D) = t + X_{t},
\end{align}
where $X_{t}$ is a random variable with probability distribution $\p_{t}$.

The differential privacy constraint \eqref{eqn:dpgeneral}    on $\KM$ is that for any $t_1,t_2 \in \R$ such that $|t_1 - t_2| \le \D $ (corresponding to the query outputs for two neighboring datasets) ,
\begin{align}
	\nm_{t_1} (S) \le e^{\e} \nm_{t_2}(S + t_1 - t_2), \forall \; \text{measurable set} \; S \subset \R,  \label{eqn:diffgeneralnoise1}
\end{align}
where for any $t \in \R$,  $S+t \, := \, \{s+t \, | \, s \in S\}$.

\subsubsection{Utility Model}

The utility model we use in this work is a very general one, which is also used in the works by Ghosh, Roughgarden, and Sundararajan \cite{Ghosh09}, Gupte and  Sundararajan \cite{minimax10}, and Brenner and Nissim \cite{Nissim10}.

Consider a cost function $\loss(\cdot): \R \to \R$, which is a function of the additive noise. Given additive noise $x$, the cost is $\loss(x)$. Given query output $t \in \R$, the additive noise is a random variable with probability distribution $\p_t$, and thus the expectation of the cost is
\begin{align}
	\int_{x\in \R} \loss(x) \p_t (dx).
\end{align}

The objective is to minimize the worst case cost among all possible query output $t \in \R$, i.e.,
\begin{align}
\text{minimize}\;  \sup_{t \in \R} \int_{x \in \R} \loss(x) \p_t(dx). \label{eqn:objective}
 \end{align}

\subsubsection{Optimization Problem}

Combining the differential privacy constraint \eqref{eqn:diffgeneralnoise1} and the objective function \eqref{eqn:objective}, we formulate a functional optimization problem:
\begin{align}
	\mathop{\text{minimize}}\limits_{ \{\p_t\}_{t \in \R}   } & \ \sup_{t \in \R} \int_{x \in \R} \loss(x) \p_t(dx) \label{eqn:objecinopt}\\
	\text{subject to} & \ \nm_{t_1} (S) \le e^{\e} \nm_{t_2}(S + t_1 - t_2), \forall \ \text{measurable set} \ S\subseteq \R, \ \forall |t_1 - t_2| \le \D.  \label{eqn:diffgeneralnoise}
\end{align}


\subsection{An Overview of Our Results}

\subsubsection{Optimal Noise Probability Distribution}

When the query output domain is the real line or the set of integers, we show (subject to some mild technical conditions on the family of differentially private mechanisms) that adding query-output-independent noise is optimal. Thus we only need to study what  the optimal noise probability distribution is. Let $\p$ denote the probability distribution of the noise added to the query output. Then the optimization problem \eqref{eqn:objecinopt} and \eqref{eqn:diffgeneralnoise} is reduced to
\begin{align}
	\mathop{\text{minimize}}\limits_{ \p} & \ \int_{x \in \R} \loss(x) \p (dx) \\
	\text{subject to} & \ \p(S) \le e^{\e} \p(S + d), \forall \ \text{measurable set} \ S\subseteq \R, \ \forall |d| \le \D.  
\end{align}

Consider a staircase-shaped probability distribution with probability density function (p.d.f.) $f_{\gamma}(\cdot)$ defined as
\begin{align}
f_{\gamma}(x)  =
\begin{cases}
  	 a(\gamma) & x \in [0, \gamma \D) \\
 e^{-\e} a(\gamma) & x \in [\gamma \D,   \D) \\
  e^{-k\e} f_{\gamma}(x - k\D)  & x \in [ k \D,   (k+1)\D) \; \text{for} \; k \in \N \\
 f_{\gamma}(-x) & x<0
  \end{cases}
\end{align}
where
\begin{align}
	a(\gamma) \triangleq \frac{ 1 - e^{-\e}}{2 \D (\gamma + e^{-\e}(1-\gamma))}
\end{align}
is a normalizing constant to make $\int_{x\in \R} f_{\gamma}(x) dx = 1$.

Our main result is
\begin{theorem}
	If the cost function $\loss(\cdot)$ is symmetric and increasing, and $\sup_{x \ge T} \frac{\loss(x+1)}{\loss(x)} < + \infty$ for some $T>0$, the optimal noise probability distribution has a staircase-shaped probability density function $f_{\gamma^*}(\cdot)$, where
	\begin{align}
		\gamma^* = \mathop{\arg \min} \limits_{\gamma \in [0,1] } \int_{x \in \R} \loss (x)  f_{\gamma}(x) dx .
	\end{align}
	
\end{theorem}

We plot the probability density functions of Laplace mechanism and staircase mechanism in Figure \ref{fig:probdf}. Figure \ref{fig:fgamma} in Section \ref{sec:result} gives a precise description of staircase mechanism.

The staircase mechanism is specified by three parameters: $\e$, $\D$, and $\gamma^*$ which is determined by $\e$ and the cost function $\loss(\cdot)$. For certain classes of cost functions, there are closed-form expressions for the optimal $\gamma^*$.

\begin{figure}[h]
\begin{subfigure}[b]{0.5\linewidth}
\centering
\includegraphics[width=0.8\textwidth]{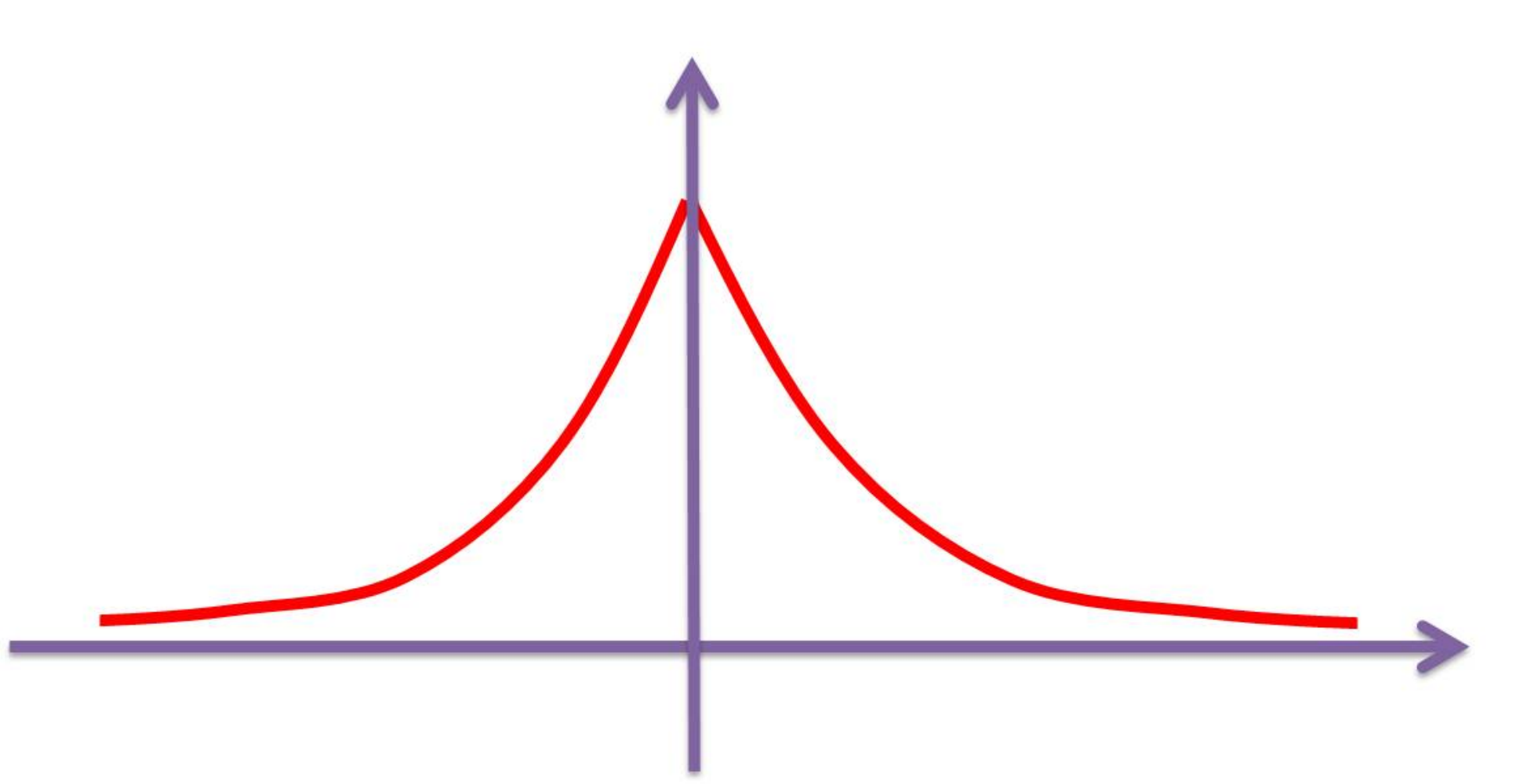}
\caption{ Laplace Mechanism }
\end{subfigure}
\begin{subfigure}[b]{0.5\linewidth}
\centering
\includegraphics[width=0.8 \textwidth]{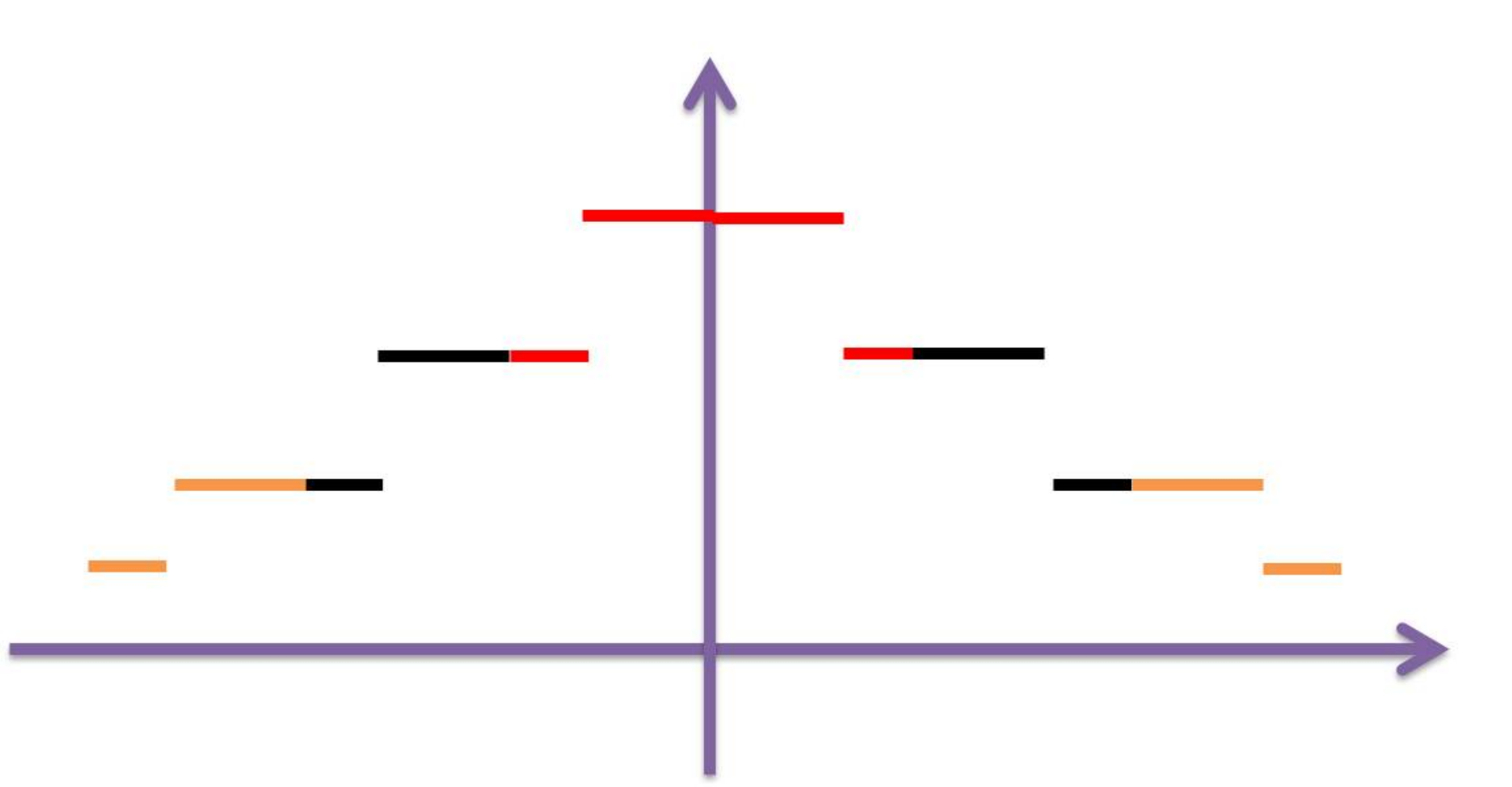}
\caption{ Staircase Mechanism}
\end{subfigure}
\caption{Probability Density Functions of Laplacian Mechanism and Staircase Mechanism}
\label{fig:probdf}
\end{figure}

\subsubsection{Applications: Minimum Noise Magnitude and Noise Power}

We apply our main result Theorem \ref{thm:main} to two typical cost functions $\loss(x) = |x|$ and $\loss(x) = x^2$, which measure noise magnitude and noise power, respectively. We derive the closed-form expressions for the optimal parameters $\gamma^*$ for these two cost functions. Comparing the optimal performances with those of the Laplacian mechanism, we show that in the high privacy regime ($\e$ is small), the Laplacian mechanism is asymptotically optimal as $\e \to 0$;  in the low privacy regime ($\e$ is large), the minimum expectation of noise amplitude  and minimum noise power are $\Theta(\D e^{-\frac{\e}{2}})$ and $\Theta(\D^2 e^{-\frac{2\e}{3}})$ as $\e \to +\infty$, while  the expectation of noise amplitude and power using the Laplacian mechanism are $\frac{\D}{\e}$ and $\frac{2\D^2}{\e^2}$, respectively, where $\D$ is the sensitivity of the query function. We conclude that the gains are more pronounced in the low privacy regime.

\subsubsection{Extension to the Discrete Setting}

Since for many important practical applications query functions are integer-valued, we also derive the optimal differentially private mechanisms for answering a single integer-valued query function.
 We show that adding query-output independent noise is optimal under a mild technical condition, and the optimal noise probability distribution has a staircase-shaped probability mass function, which can be viewed as the discrete variant of the staircase mechanism in the continuous setting.

 This result helps us directly compare our work and the existing works \cite{Ghosh09, minimax10} on integer-valued query functions. Our result shows that for integer-valued query function, the optimal noise probability mass function is also  staircase-shaped, and in the case the sensitivity $\D = 1$, the optimal probability mass function is reduced to the geometric distribution, which was derived in \cite{Ghosh09, minimax10}. Therefore, this result can be viewed as a generalization of \cite{Ghosh09, minimax10} in the discrete setting for query functions with arbitrary sensitivity.




\subsection{Connection to the Literature}\label{subsec:connection}

In this section, we discuss the relations of our results and some directly related works in the literature, and the implications of our results on  other works.

\subsubsection{Laplacian Mechanism vs Staircase Mechanism}

The Laplacian mechanism is specified by two parameters, $\e$ and the query function sensitivity $\D$. $\e$ and $\D$ completely characterize the  differential privacy constraint. On the other hand, the staircase mechanism is specified by three parameters, $\e$, $\D$, and $\gamma^*$ which is determined by $\e$ and the utility function/cost function. For certain classes of utility functions/cost  functions, there are closed-form expressions for the optimal $\gamma^*$.

From the two examples given in Section \ref{sec:application}, we can see that although the Laplacian mechanism is not strictly optimal, in the high privacy regime ($\e \to 0$), Laplacian mechanism is asymptotically optimal:
\begin{itemize}
	\item For the expectation of noise amplitude, the additive gap from the optimal values goes to 0 as $\e \to 0$,
	\item For noise power, the additive gap from the optimal values is upper bounded by a constant  as $\e \to 0$.
\end{itemize}
 However, in the low privacy regime ($\e \to +\infty$), the multiplicative gap from the optimal values can be arbitrarily large. We conclude that in the high privacy regime, the Laplacian mechanism is nearly optimal, while in the low privacy regime significant improvement can be achieved by using the staircase mechanism. We plot the multiplicative gain of staircase mechanism over Laplacian mechanism for expectation of noise amplitude and noise power in Figure  \ref{fig:comparison}, where $V_{\text{Optimal}}$ is the optimal (minimum) cost, which is achieved by staircase mechanism, and $V_{Lap}$ is the cost of Laplacian mechanism. We can see that for $\e \approx 10$, the staircase mechanism has about 15-fold and 23-fold improvement, with noise amplitude and power respectively. While $\epsilon \approx 10$ corresponds to really low privacy, our results show that low privacy can be had very cheaply (particularly when compared to the state of the art Laplacian mechanism).

Since the staircase mechanism is derived under the same problem setting as Laplacian mechanism, the staircase mechanism can be applied {\em wherever} Laplacian mechanism is used, and it performs {\em strictly better} than Laplacian mechanism (and significantly better in low privacy scenarios).


\begin{figure}[h]
\begin{subfigure}[b]{0.5\linewidth}
\centering
\includegraphics[width=1.1\textwidth]{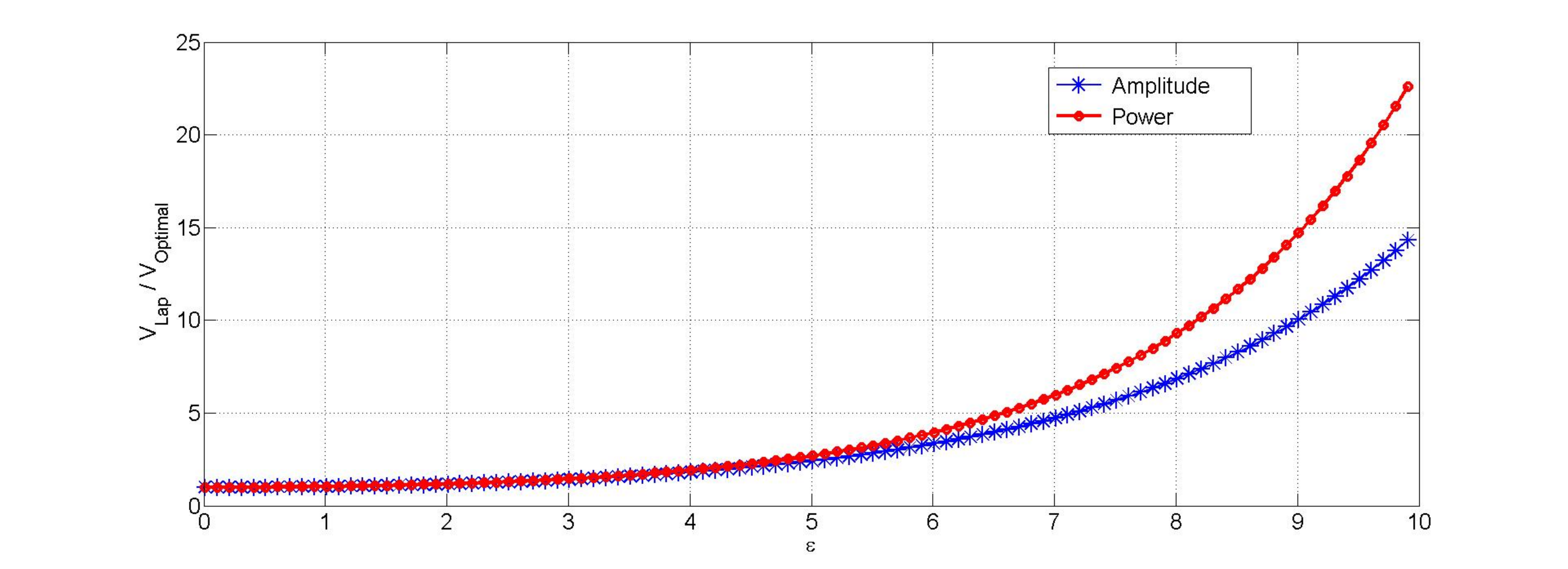}
\caption{ $0 < \e \le 10 $}
\label{fig:compare1}
\end{subfigure}
\begin{subfigure}[b]{0.5\linewidth}
\centering
\includegraphics[width=1.1\textwidth]{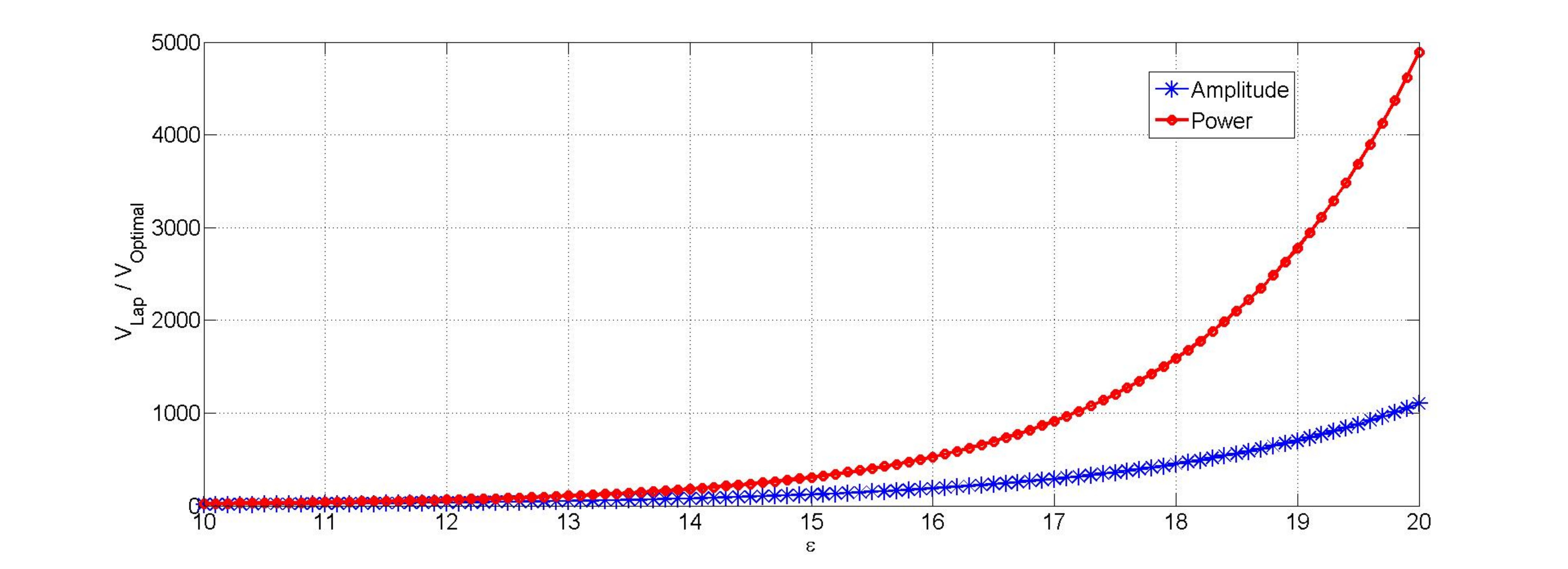}
\caption{ $10 \le \e \le 20$}
\label{fig:compare2}
\end{subfigure}
\caption{Multiplicative Gain of the Staircase Mechanism over the Laplacian Mechanism. }
\label{fig:comparison}
\end{figure}

\subsubsection{Relation to Shamai and Verdu, \cite{SV92}}
Shamai and Verdu \cite{SV92} consider the minimum variance noise for a fixed value of the average of false alarm and missed detection probabilities of binary hypothesis testing. In \cite{SV92}, the binary hypotheses correspond to the signal being in a binary set $\{ -\Delta, +\Delta\}$. Their solution involved the noise being discrete and, further, having a pmf on the integer lattice (scaled by $\Delta$). Our setting is related, but is differentiated via the following two key distinctions:
\begin{itemize}
\item Instead of a constraint on the sum of false alarm and missed detection probabilities, we have constraints on symmetric weighted combinations of the two error probabilities (as in Equations~\eqref{eq:famd1} and~\eqref{eq:famd2}).
\item Instead of the binary hypotheses corresponding to the signal being in a binary set $\{-\Delta, +\Delta\}$ we consider all possible binary hypotheses for the signal to be in $\{ x_1, x_2\}$ where $x_1, x_2 \in [-\Delta, \Delta]$ are arbitrtary.
\end{itemize}

\subsubsection{Relation to Ghosh et. al. \cite{Ghosh09} }

Ghosh, Roughgarden, and Sundararajan  \cite{Ghosh09} show that for a single count query with sensitivity $\D = 1$, for a general class of utility functions, to minimize the expected cost under a Bayesian framework the optimal mechanism to preserve differential privacy is the geometric mechanism, which adds noise with geometric distribution.

We discuss the relations and differences between \cite{Ghosh09} and our work in the following: Both \cite{Ghosh09} and our work are similar in that,  given the query output, the cost function only depends on the additive noise magnitude, and is an increasing function of noise magnitude. On the other hand, there are two main differences:
\begin{itemize}

	\item \cite{Ghosh09}  works under a Bayesian setting, while ours is to minimize the worst case cost.

	\item \cite{Ghosh09} studies a count query where the query output is integer-valued, bounded and sensitivity is unity.  In our work,  we first study general real-valued query function where the query output can take any real value, and then generalize the result to discrete setting where query output is integer valued.  In both cases, the sensitivity of query functions can be arbitrary, not restricted to one.
\end{itemize}

\subsubsection{Relation to  Gupte and Sundararajan  \cite{minimax10}  }

  Gupte and Sundararajan \cite{minimax10} derive the optimal noise probability distributions for a single count query with sensitivity $\D = 1$ for minimax (risk-averse) users. Their model is the same as the one in \cite{Ghosh09} except that their objective function is to minimize the worst case cost, the same as our objective.  \cite{minimax10} shows that although there is no universally optimal solution to the minimax optimization problem in \cite{minimax10} for a general class of cost functions, each solution (corresponding to different cost functions) can be derived from the same geometric mechanism by randomly remapping.

  As in \cite{Ghosh09}, \cite{minimax10} assumes the query-output is bounded. Our result shows that when the query sensitivity is one, without any boundedness knowledge about the query-output, the optimal mechanism is to add random noise with geometric distribution to the query output.

\subsubsection{Relation to Brenner and Nissim  \cite{Nissim10}  }

While \cite{Ghosh09} shows that  for a single count query with sensitivity $\D = 1$, there is a universally optimal mechanism for a general class of utility functions  under a Bayesian framework, Brenner and Nissim  \cite{Nissim10}  show that for general query functions no universally optimal mechanisms exist. Indeed, this follows directly from our results:  under our optimization framework, the optimal mechanism is adding noise with staircase-shaped probability distribution which is specified by three parameters $\e, \D$ and $\gamma^*$, where in general $\gamma^*$ depends on the cost function.  Generally, for different cost functions, the optimal noise probability distributions have staircase-shaped probability density functions specified by different parameters.

\subsubsection{Relation to   Nissim, Raskhodnikova and Smith \cite{NRS07}  }

Nissim, Raskhodnikova and Smith \cite{NRS07} show that for certain nonlinear query functions, one can improve the accuracy by adding data-dependent noise calibrated to the smooth sensitivity of the query function, which is based on the local sensitivity of the query function. In our model, we use   the global sensitivity of the query function only, and assume that the local sensitivity is the same as the global sensitivity, which holds for a general class of query functions, e.g., count, sum.

\subsubsection{Relation to Hardt and Talwar  \cite{geometry}  }

Hardt and Talwar \cite{geometry} study the tradeoff between privacy and error for answering a set of linear queries over a histogram in a differentially private way. The  error is defined as the worst  expectation of the $\ell^2$-norm of the noise. The lower bound given in \cite{geometry} is $\Omega( \e^{-1} d\sqrt{d})$, where $d$ is the number of linear queries. An immediate consequence of our result is that for fixed d, when $\e \to +\infty$, an upper bound of $\Theta(e^{-\frac{\e}{3d}}d \sqrt{d})$ is achievable by adding independent staircase-shaped noise with parameter $\frac{\e}{d}$ to each component.

\subsubsection{Relation to Other Works}

There are many existing works on studying how to improve the accuracy for answering more complex queries under differential privacy, in which the basic building block is the standard Laplacian mechanism. For example, Hay et. al. \cite{Hay10} show that  one can improve the accuracy for a general class of histogram queries, by exploiting the consistency constraints on the query output, and Li et. al. \cite{Li10}  study how to optimize linear counting queries under differential privacy by carefully choosing the set of linear queries to be answered. In these works, the error is measured by the mean squared error of query output estimates, which corresponds to the variance of the noise added to the query output to preserve differential privacy. In terms of $\epsilon$, the error bound in these works scales linearly to $\frac{1}{\e^2}$, because of the use of Laplacian noise. If Laplacian distribution is replaced by staircase distribution in these works, one can improve the error bound to $\Theta(e^{-C\e})$ (for some constant $C$ which depends on the number of queries) when $\e \to +\infty$ (corresponding to the low privacy regime).



\subsection{Organization}

The paper is organized as follows. We show the optimality of query-output independent perturbation in Section \ref{sec:optimality}, and present the optimal differentially private mechanism, staircase mechanism, in Section \ref{sec:result}. In Section \ref{sec:application}, we apply our main result to derive the optimal noise probability distribution with minimum expectation of noise amplitude and power, respectively, and  compare the performances with the Laplacian mechanism.  Section \ref{sec:gammaproperty} presents the asymptotic properties of $\gamma^*$ in the staircase mechanism for momentum cost functions, and suggests a heuristic choice of $\gamma$ that appears to work well for a wide class of cost functions. Section \ref{sec:discrete} generalizes the staircase mechanism for integer-valued query function in the discrete setting, and Section \ref{sec:abstractsetting} extends the staircase mechanism to the abstract setting. Section \ref{sec:conclusion} concludes this paper.


\section{Optimality of Query-Qutput Independent Perturbation} \label{sec:optimality}

Recall that the optimization problem we study in this work is
\begin{align}
	\mathop{\text{minimize}}\limits_{ \{\p_t\}_{t \in \R}   } & \ \sup_{t \in \R} \int_{x \in \R} \loss(x) \p_t(dx) \label{eqn:objecinopt1}\\
	\text{subject to} & \ \nm_{t_1} (S) \le e^{\e} \nm_{t_2}(S + t_1 - t_2), \forall \ \text{measurable set} \ S\subseteq \R, \ \forall |t_1 - t_2| \le \D,  \label{eqn:dpmeasure1}
\end{align}
where $\p_t$ is the noise probability distribution when the query output is $t$.

Our claim  is that in the optimal family of probability distributions, $\nm_{t}$ can be  independent of $t$, i.e., the probability distribution of noise is independent of the query output. We  prove this claim under a  technical condition which assumes that $\{\nm_{t}\}_{t \in \R}$ is piecewise constant and periodic (the period  can be arbitrary) in terms of $t$.

For any positive integer $n$, and for any positive real number $T$, define
\begin{align}
 	\KM_{T,n} \triangleq \{\; \{\nm_t\}_{t \in \R}  \, |\, \{\nm_t\}_{t \in \R} \; \text{satisfies} \; \eqref{eqn:diffgeneralnoise}, \;  & \nm_t = \nm_{k \frac{T}{n}},   \text{for}\;   t \in [k\frac{T}{n}, (k+1)\frac{T}{n}), k \in \Z, \nonumber  \\
 & \text{and} \; \nm_{t+T} = \nm_{t}, \forall t \in \R    \}.
 \end{align}

\begin{theorem} \label{thm:maingeneral}
 	Given any family of probability distribution $\{\nm_t\}_{t \in \R} \in  \cup_{T >0} \cup_{n \ge 1} \KM_{T,n}$,  there exists a probability distribution $\nm^*$ such that the family of probability distributions  $\{\p^*_t\}_{t \in \R}$ with $\p^*_t \equiv \p^*$  satisfies the differential privacy constraint \eqref{eqn:diffgeneralnoise} and
 	\begin{align}
 		  \sup_{t \in \R} \int_{x \in \R} \loss(x) \p^*_t(dx)     \le   \sup_{t \in \R} \int_{x \in \R} \loss(x) \p_t(dx) .
 	\end{align}	
 \end{theorem}

\begin{proof}
  Here we briefly discuss the main proof technique. For complete proof, see  Appendix \ref{app:optimality}. The proof of Theorem \ref{thm:maingeneral} uses two properties on the family of probability distributions satisfying differential privacy constraint \eqref{eqn:diffgeneralnoise}. First, we show that for any family of  probability distributions satisfying \eqref{eqn:diffgeneralnoise}, any translation of the probability distributions will also preserve differential privacy, and the cost is the same.  Second, we show that given a collection of families of probability distributions each of which satisfies \eqref{eqn:diffgeneralnoise}, we can take a convex combination of them to construct a new family of probability distributions satisfying \eqref{eqn:diffgeneralnoise} and the new cost is not worse.  Due to these two properties, given any family of probability distributions    $\{\nm_t\}_{t \in \R}  \in \cup_{T >0} \cup_{n \ge 1} \KM_{T,n}$, one can take a convex combination of different translations of $ \{\nm_t\}_{t \in \R}$ to construct $\{\p^*_t\}_{t \in \R}$ with $\p^*_t \equiv \p^*$, and the cost is not worse. 
\end{proof}

Theorem \ref{thm:maingeneral} states that if we assume the family of noise probability distributions is piecewise constant (over intervals with length $\frac{T}{n}$) in terms of $t$, and periodic over $t$ (with period $T$), where $T,n$ can be arbitrary, then in the optimal mechanism we can assume $\p_t$ does not dependent on $t$. We conjecture that the technical condition can be done away with.

\section{Optimal Noise Probability distribution} \label{sec:result}

Due to Theorem \ref{thm:maingeneral}, to derive the optimal randomized mechanism to preserve differential privacy, we can restrict to  noise-adding mechanisms where the noise probability distribution does not depend on the query output. In this section we state our main result Theorem \ref{thm:main} on the optimal noise probability distribution.

Let $\p$ denote the probability distribution of the noise added to the query output. Then the optimization problem  in \eqref{eqn:objecinopt} and \eqref{eqn:diffgeneralnoise} is reduced to
\begin{align}
	\mathop{\text{minimize}}\limits_{ \p} & \ \int_{x \in \R} \loss(x) \p (dx) \\
	\text{subject to} & \ \p(S) \le e^{\e} \p(S + d), \forall \ \text{measurable set} \ S\subseteq \R, \ \forall |d| \le \D.  \label{eqn:dpconstraintfinal}
\end{align}

We assume that the cost function $\loss(\cdot)$ satisfies two (natural) properties.

\begin{property}\label{property1}
	$\loss(x)$ is a symmetric function, and monotonically increasing for $x \ge 0$, i.e, $\loss(x)$ satisfies
	\begin{align}
		\loss(x) &= \loss(-x), \forall x \in \R,
	\end{align}
	and
	\begin{align}
		\loss(x) &\le \loss(y), \forall 0 \le x \le y.
	\end{align}
\end{property}

In addition, we assume $\loss(x)$ satisfies a mild technical condition which essentially says that $\loss(\cdot)$ does not increase too fast (while still allowing it to be unbounded).


\begin{property}\label{property2}
	There exists a positive integer $T$ such that $\loss(T) >0$ and $\loss(x)$ satisfies
	\begin{align}
	 	\sup_{x \ge T} \frac{\loss(x+1)}{\loss(x)} < + \infty. \label{eqn:propertyL}
	 \end{align}
\end{property}

\begin{figure}[t]
\centering
\includegraphics[scale=0.6]{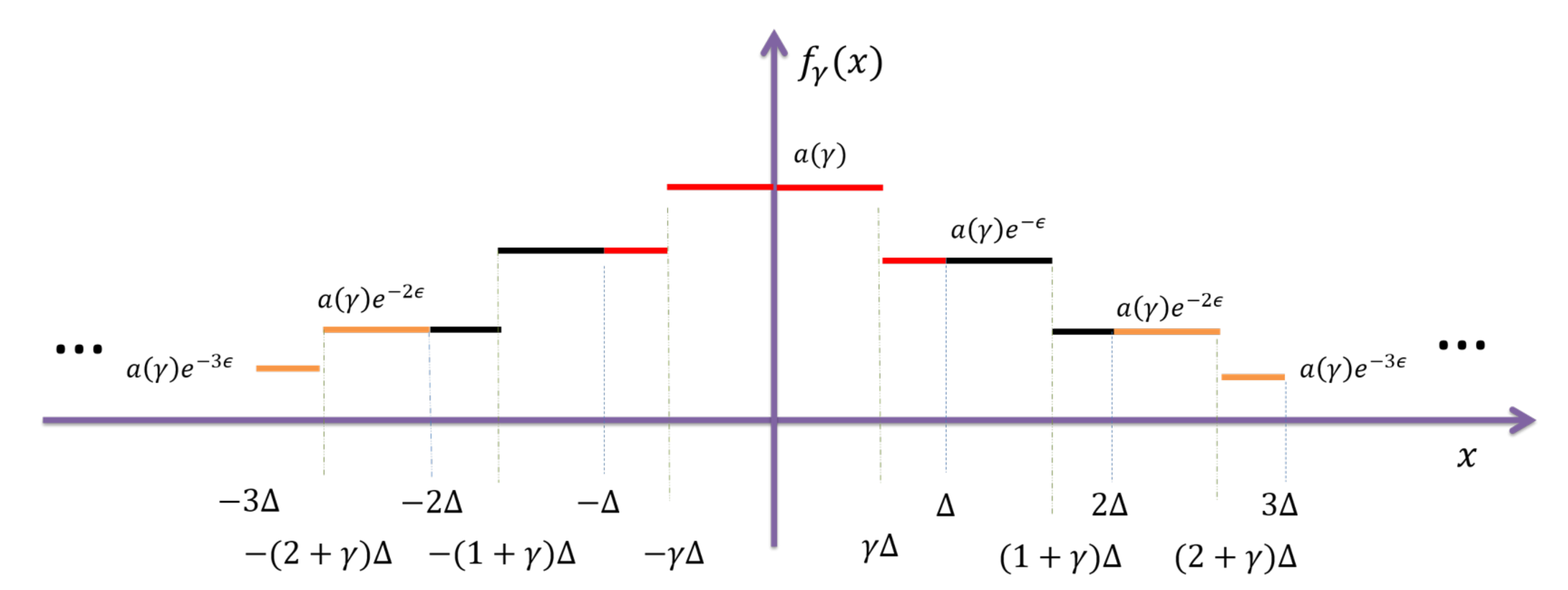}
\caption{The Staircase-Shaped Probability Density Function $f_{\gamma}(x)$}
\label{fig:fgamma}
\end{figure}

Consider a staircase-shaped probability distribution with probability density function (p.d.f.) $f_{\gamma}(\cdot)$ defined as
\begin{align}
f_{\gamma}(x)  =
\begin{cases}
  	 a(\gamma) & x \in [0, \gamma \D) \\
 e^{-\e} a(\gamma) & x \in [\gamma \D,   \D) \\
  e^{-k\e} f_{\gamma}(x - k\D)  & x \in [ k \D,   (k+1)\D) \; \text{for} \; k \in \N \\
 f_{\gamma}(-x) & x<0
  \end{cases}\label{eqn:deffgamma}
\end{align}
where
\begin{align}
	a(\gamma) \triangleq \frac{ 1 - e^{-\e}}{2 \D (\gamma + e^{-\e}(1-\gamma))}
\end{align}
is a normalizing constant to make $\int_{x\in \R} f_{\gamma}(x) dx = 1$.
It is easy to check that for any $\gamma \in [0,1]$, the probability distribution with p.d.f. $f_{\gamma}(\cdot)$ satisfies the differential privacy constraint \eqref{eqn:dpconstraintfinal}. Indeed, the probability density function $f_{\gamma}(x)$ satisfies
\begin{align}
	f_{\gamma}(x) \le e^{\e} f_{\gamma}(x + d) , \forall x \in \R, |d| \le \D,
\end{align}
which implies \eqref{eqn:dpconstraintfinal}.

Let $\sP$ denote the set of all probability distributions satisfying \eqref{eqn:dpconstraintfinal}.
Our main result on the optimal noise probability distribution is:
\begin{theorem}\label{thm:main}
If the cost function $\loss(x)$ satisfies Property \ref{property1} and Property \ref{property2}, then
\begin{align}
	\inf_{\p \in \sP}  \int_{x \in \R} \loss (x)  \p(dx)  = \inf_{\gamma \in [0,1]}  \int_{x \in \R} \loss (x)  f_{\gamma}(x) dx .
\end{align}
\end{theorem}

\begin{proof}
	 Here we briefly discuss the main proof idea and technique. First, by deriving several properties on the probability distributions satisfying the $\e$-differential privacy constraint, we show that without loss of generality, one can ``discretize'' any valid probability distribution, even for those which do not have probability density functions. Second, we show that to minimize the cost, the probability density function of the discretized probability distribution should be monotonically and geometrically decaying. Lastly, we show that the optimal probability density function should be staircase-shaped. For the complete proof, see Appendix \ref{sec:proof}.
\end{proof}

Therefore, the optimal noise probability distribution to preserve $\e$-differential privacy for any real-valued query function has a staircase-shaped probability density function, which is specified by three parameters $\e$, $\D$ and $\gamma^* = \mathop{\arg \min} \limits_{\gamma \in [0,1] } \int_{x \in \R} \loss (x)  f_{\gamma}(x) dx $.

A natural and simple algorithm to generate random noise with staircase distribution is given in Algorithm \ref{algo:staircase_mech}.




\begin{algorithm}
\caption{Generation of Random Variable with Staircase Distribution}
\label{algo:staircase_mech}
\begin{algorithmic}
\State \textbf{Input: } $\e$, $\D$, and  $\gamma \in [0,1]$.
\State \textbf{Output: } $\noise$, a random variable (r.v.) with staircase distribution specified by $\e, \D$ and $\gamma$.
\\
\State Generate a r.v. $S$ with $\text{Pr}[S = 1] = \text{Pr}[S = -1] = \frac{1}{2}$.
\State Generate a geometric r.v. $G$ with $\text{Pr}[G = i] = (1-b)b^i $ for integer $i \ge 0$, where $b = e^{-\e}$.
\State Generate a r.v. $U$ uniformly distributed in $[0,1]$.
\State Generate a binary r.v. $B$ with $\text{Pr}[B = 0] = \frac{\gamma}{\gamma + (1-\gamma)b}$ and $\text{Pr}[B = 1] = \frac{(1-\gamma)b}{\gamma + (1-\gamma)b}$.
\State $\noise \gets S \left( (1-B)\left((G+\gamma U)\D\right) + B\left((G+\gamma + (1-\gamma)U)\D\right)\right)$.
\State Output $\noise$.
\end{algorithmic}
\end{algorithm}

In the formula,
\begin{align}
	\noise \gets S \left( (1-B)\left((G+\gamma U)\D\right) + B\left((G+\gamma + (1-\gamma)U)\D\right) \right),
\end{align}
\begin{itemize}
	\item $S$ determines the sign of the noise,
	\item $G$ determines which interval $[G\D, (G+1)\D)$ the noise lies in,
	\item $B$ determines which subinterval of $[G\D, (G+\gamma)\D)$ and $[(G+\gamma)\D, (G+1)\D)$ the noise lies in,
	\item $U$ helps to uniformly sample  the subinterval.
\end{itemize}

\section{Applications} \label{sec:application}
In this section, we apply our main result Theorem \ref{thm:main} to derive the parameter $\gamma^*$ of the staircase mechanism with minimum expectation of noise magnitude  and noise second moment, respectively, and then compare the performances with the Laplacian mechanism.

\subsection{Optimal Noise Probability Distribution with Minimum Expectation of Noise Amplitude}

To minimize the expectation of amplitude, we have cost function $\loss(x) = |x|$, and it is easy to see that it satisfies Property \ref{property1} and Property \ref{property2}.

To simplify notation, define $b \triangleq e^{-\e}$, and define
\begin{align}
	V(\p) \triangleq \int_{x \in \R} \loss(x)  \p(dx).
\end{align}
for a given probability distribution $\p$.

\begin{theorem}\label{thm:1}
	To minimize the expectation of the amplitude of noise, the optimal noise probability distribution has probability density function $f_{\gamma^*}(\cdot)$ with
\begin{align}
	\gamma^* = \frac{1}{1 + e^{\frac{\e}{2}}},
\end{align}
and the minimum expectation of noise amplitude   is
\begin{align}
	V(\p_{\gamma^*}) = \D \frac{e^{\frac{\e}{2}}}{e^{\e}-1}.
\end{align}

\end{theorem}

\begin{proof}
	See Appendix \ref{app:1}.
\end{proof}

\begin{figure}[t]
\centering
\includegraphics[scale=0.2]{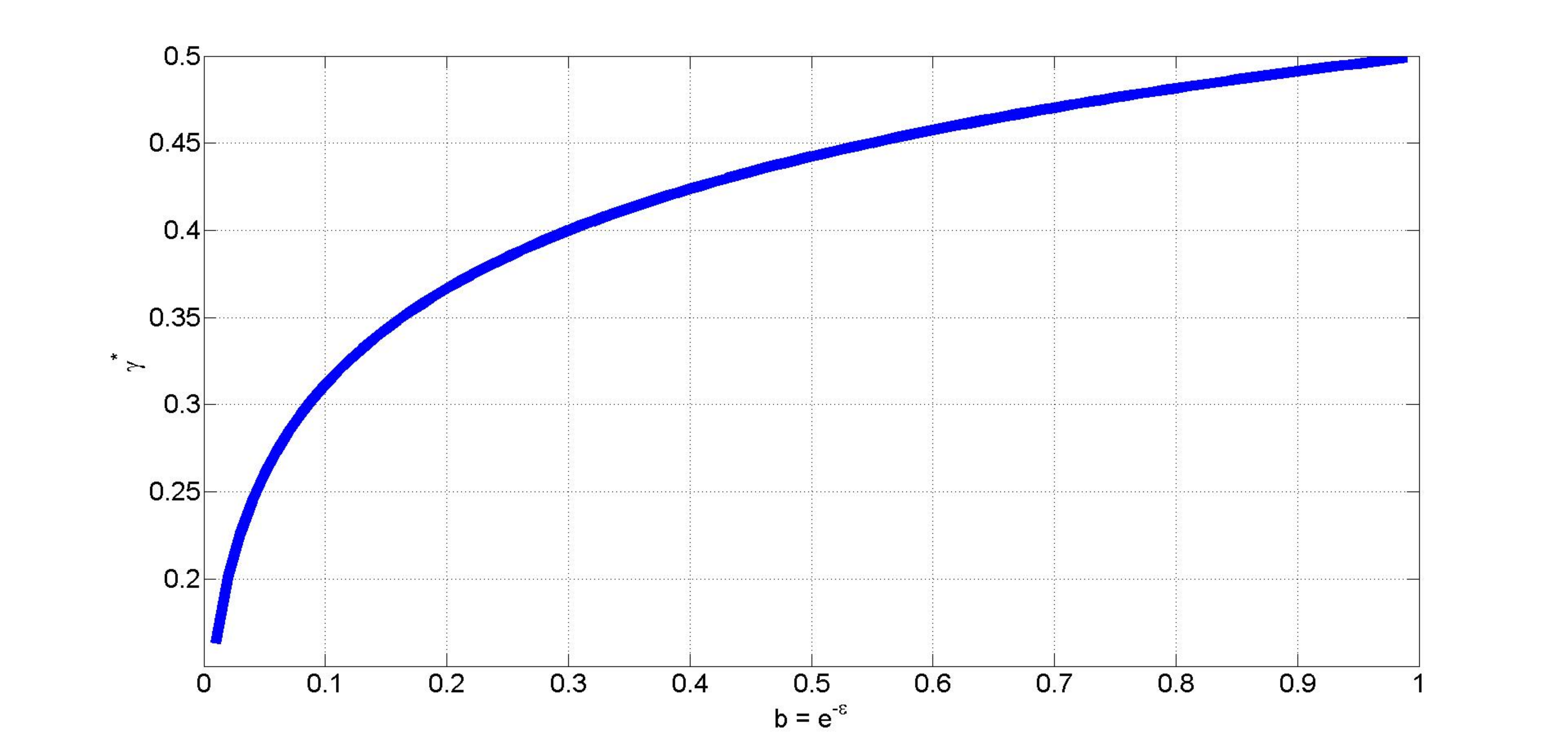}
\caption{Optimal $\gamma^*$ for cost function $L(x) = x^2$}
\label{fig:optimalr}
\end{figure}

Next, we compare the performances of the optimal noise probability distribution and the Laplacian mechanism.  The Laplace distribution has probability density function
\begin{align}
	f(x) = \frac{1}{2\lambda} e^{-\frac{|x|}{\lambda}},
\end{align}
where $\lambda = \frac{\D}{\e}$. So the expectation of the amplitude of noise with Laplace distribution is
\begin{align}
	V_{Lap} \triangleq \int_{-\infty}^{+\infty} |x|f(x) dx = \frac{\D}{\e}.
\end{align}

By comparing $V(\p_{\gamma^*})$ and $V_{Lap}$, it is easy to see  that in the high privacy regime ($\e$ is small) Laplacian mechanism is asymptotically optimal, and the additive gap from optimal value goes to 0 as $\e \to 0$; in the low privacy regime ($\e$ is large),  $V_{Lap} = \frac{\D}{\e})$, while  $V(\p_{\gamma^*}) =  \Theta(\D e^{-\frac{\e}{2}})$. Indeed,
\begin{corollary}\label{thm:2}
	Consider the cost function $\loss(x) = |x|$. In the high privacy regime ($\e$ is small),
	\begin{align}
		V_{Lap} - V(\p_{\gamma^*}) &= \D\left( \frac{\e}{24} - \frac{7\e^3}{5760} + O(\e^5) \right),
	\end{align}
	as $\e \to 0$.

	And in the low privacy regime ($\e$ is large),
	\begin{align}
		V_{Lap} &= \frac{\D}{\e}, \\
		V(\p_{\gamma^*}) &=  \Theta(\D e^{-\frac{\e}{2}}),
	\end{align}
	as $\e \to +\infty$.
\end{corollary}

\subsection{Optimal Noise Probability Distribution with Minimum Power}

Given the probability distribution $\p$ of the noise, the power of noise is defined as $\int_{x \in \R} x^2 \p (dx)$. Accordingly, the cost function   $\loss(x) = x^2 $,  and it is easy to see it satisfies Property \ref{property1} and Property \ref{property2}.

Recall $b \triangleq e^{-\e}$.

\begin{theorem}\label{thm:3}
	To minimize the power of noise (accordingly, $\loss(x) = x^2$), the optimal noise probability distribution has probability density function $f_{\gamma^*}(\cdot)$ with
\begin{align}
	\gamma^* = -\frac{b}{1-b} + \frac{(b-2b^2+2b^4-b^5)^{1/3}}{2^{1/3} (1-b)^2},
\end{align}
and the minimum power of noise is
\begin{align}
	V(\p_{\gamma^*}) = \D^2\frac{ 2^{-2/3} b^{2/3} (1+b)^{2/3} + b }{ (1-b)^2 }.
\end{align}

\end{theorem}

\begin{proof}
	See Appendix \ref{app:3}.
\end{proof}

Next, we compare the performances of the optimal noise probability distribution and the Laplacian mechanism.  The power of noise with Laplace distribution with $\lambda = \frac{\D}{\e}$ is
\begin{align}
	V_{Lap} \triangleq \int_{-\infty}^{+\infty} x^2 \frac{1}{2\lambda} e^{-\frac{|x|}{\lambda}} dx = 2 \frac{\D^2}{\e^2}.
\end{align}

By comparing $V(\p_{\gamma^*})$ and $V_{Lap}$, it is easy to see  that in the high privacy regime ($\e$ is small) Laplacian mechanism is asymptotically optimal, and the additive gap from optimal value is upper bounded by a constant as $\e \to 0$; in the low privacy regime ($\e$ is large),  $V_{Lap} = \Theta(\frac{2\D^2}{\e^2})$, while  $V(\p_{\gamma^*}) =  \Theta(\D^2 e^{-\frac{2\e}{3}})$. Indeed,
\begin{corollary}\label{thm:4}
	Consider the cost function $\loss(x) = x^2$. In the high privacy regime ($\e$ is small),
	\begin{align}
		   V_{Lap}- V(\p_{\gamma^*}) &= \D^2 \left( \frac{1}{12} - \frac{\e^2}{720} + O(\e^4) \right),
	\end{align}
	as $\e \to 0$.

	And in the low privacy regime ($\e$ is large),
	\begin{align}
		V_{Lap} &= \frac{2\D^2}{\e^2}, \\
		V(\p_{\gamma^*}) &=  \Theta(\D^2 e^{-\frac{2\e}{3}}),
	\end{align}
	as $\e \to +\infty$.
\end{corollary}





\section{Property of $\gamma^*$} \label{sec:gammaproperty}

In this section, we derive some asymptotic properties of the optimal $\gamma^*$ for moment cost functions, and give a heuristic choice of $\gamma$ which only depends on $\e$.

\subsection{Asymptotic Properties of $\gamma^*$}

In Section \ref{sec:application}, we have seen that for the cost functions $\loss(x) = |x| $ and $\loss(x) = x^2$, the optimal $\gamma^*$ lies in the interval $[0,\frac{1}{2}]$ for all $\e$ and is a monotonically decreasing function of $\e$; and furthermore, $\gamma^* \to \frac{1}{2}$ as $\e$ goes to $0$, and $\gamma^* \to 0$ as $\e$ goes to $+\infty$.

We generalize these asymptotic properties of $\gamma$ as a function of $\e$ to all moment cost functions. More precisely, given $m \in \N$ and $m \ge 1$,

\begin{theorem} \label{thm:gammaprop}
	Consider the cost function $\loss(x) = |x|^m$. Let $\gamma^*$ be the optimal $\gamma$ in the staircase mechanism for $\loss(x)$, i.e.,
	\begin{align}
		\gamma^* = \mathop{\arg \min} \limits_{\gamma \in [0,1] } \int_{x \in \R} |x|^m  f_{\gamma}(x) dx.
	\end{align}
	We have
	\begin{align}
		\gamma^* &\to \frac{1}{2}, \; \text{as} \; \e \to 0, \\
		\gamma^* &\to 0,  \; \text{as}\;  \e \to +\infty.
	\end{align}
\end{theorem}

\begin{proof}
	See Appendix \ref{app:gammaprop}.
\end{proof}

\begin{corollary}
	For all the cost functions $\loss(\cdot)$ which can be written as 
\begin{align}
	\loss(x) = \sum_{i=1}^n \alpha_i |x|^{d_i},
\end{align}
where $n \ge 1$, $\alpha_i \in \R, d_i \in \N$ and $\alpha_i,d_i \ge 0$ for all $i$,  the optimal $\gamma^*$ in the staircase mechanism has the following asymptotic properties:
	\begin{align}
		\gamma^* &\to \frac{1}{2}, \; \text{as} \; \e \to 0, \\
		\gamma^* &\to 0,  \; \text{as}\;  \e \to +\infty.
	\end{align}
\end{corollary}


\subsection{A Heuristic Choice of $\gamma$}

We have shown that in general the optimal $\gamma^*$ in the staircase mechanism depends on both $\e$ and the cost function $\loss(\cdot)$. Here we give a heuristic choice  of $\gamma$ which depends only on $\e$, and show that the performance is reasonably good in the low privacy regime.

Consider a particular choice of $\gamma$, which is
\begin{align}
	\gammaheu := \frac{b}{2} = \frac{e^{-\e}}{2}.
\end{align}

It is easy to see that $\gammaheu$ has the same asymptotic properties as the optimal $\gamma^*$ for momentum cost functions, i.e.,
\begin{align}
	\gammaheu \to 0, \; \text{as} \;  b \to 0, \\
	\gammaheu \to \frac{1}{2}, \; \text{as} \;  b \to 1.
\end{align}

Furthermore, the probability that the noise magnitude is less than $\frac{e^{-\e}}{2}\D$ is approximately $\frac{1}{3}$ in the low privacy regime ($\e \to +\infty$). Indeed,
\begin{align}
	\text{Pr}[ |\noise| \le \frac{e^{-\e}}{2} \D] = \text{Pr}[ |\noise| \le \gammaheu \D]
	= 2 a(\gammaheu) \gammaheu \D
	= \frac{1-b}{\gammaheu + b(1-\gammaheu)} \gammaheu
	= \frac{b-b^2}{3b-b^2},
\end{align}
which goes to $\frac{1}{3}$ as $\e \to +\infty$ (accordingly,  $b \to 0$).

On the other hand, for Laplace mechanism,
\begin{align}
	\text{Pr}[ |\noise| \le \frac{e^{-\e}}{2} \D] &= \int_{-\frac{e^{-\e}}{2} \D}^{\frac{e^{-\e}}{2} \D} \frac{1}{2\lambda} e^{-\frac{|x|}{\lambda}}dx
	= 1 - e^{-\frac{\e e^{-\e}}{2}},
\end{align}
which goes to zero as $\e \to +\infty$.

We conclude that in the low privacy regime as $\e \to +\infty$, the staircase mechanism with the heuristic parameter $\gammaheu  = \frac{e^{-\e}}{2}$ can guarantee with probability about $\frac{1}{3}$ the additive noise is very close to zero, while the probability given by Laplacian mechanism is approximately zero.

\section{Extension to The Discrete Setting} \label{sec:discrete}

In this section, we extend our main result Theorem \ref{thm:maingeneral} and Theorem \ref{thm:main} to the discrete settings, and show that the optimal noise-adding mechanism in the discrete setting is a discrete variant of the staircase mechanism in the continuous setting.

\subsection{Problem Formulation}

We first give the problem formulation in the discrete setting.

Consider an integer-valued query function \footnote{Without loss of generality, we assume that in the discrete setting the query output is integer-valued. Indeed, any uniformly-spaced discrete setting can be reduced to the integer-valued setting by scaling the query output.  }
\begin{align}
	q: \database \rightarrow \Z,
\end{align}
where $\database$ is the domain of the databases. Let $\D$ denote the sensitivity of the query function $q$ as defined in \eqref{def:sensitivity}. Clearly, $\D$ is an integer in this discrete setting.

In the discrete setting, a generic randomized mechanism $\KM$ is a family of noise probability distributions over the domain $\Z$ indexed by the query output (denoted by $i$), i.e.,
\begin{align}
	\KM = \{ \p_i : i \in \Z \},
\end{align}
and given dataset $D$, the mechanism $\KM$ will release the query output $i = q(D)$ corrupted by additive random noise with probability distribution  $\p_i$:
\begin{align}
    \KM(D) = i + X_{i},
\end{align}
where $X_{i}$ is a discrete random variable with probability distribution $\p_{i}$.

Then, the $\e$-differential privacy constraint \eqref{eqn:dpgeneral} on $\KM$ is that for any $i_1, i_2 \in Z$ such that $ |i_1 - i_2| \le \D$ (corresponding to the query outputs for two neighboring datasets), and for any subset $ S \subset \Z$,
\begin{align}
	\nm_{i_1} (j) \le e^{\e} \nm_{i_2}(j + i_1 - i_2), \forall j \in \Z, |i_1 - i_2| \le \D,  \label{eqn:disgeneralnoise}
\end{align}
and the  goal is to minimize the worst-case cost
\begin{align}
	\sup_{i \in \Z} \sum_{j=-\infty}^{+\infty} \loss(j) \nm_i(j)
\end{align}
subject to the differential privacy constraint  \eqref{eqn:disgeneralnoise}.

\subsection{Optimality of Query-Qutput Independent Perturbation}

In this section, we show that query-output independent perturbation is optimal in the discrete setting.

For any integer $n \ge 1$, define
\begin{align}
 	\KM_{n} \triangleq \{ \; \{\nm_i\}_{i \in \Z}  | \; \{\nm_i\}_{i \in \Z} \; \text{satisfies}\;  \eqref{eqn:disgeneralnoise},\;  \text{and} \;  \nm_{i+n} = \nm_{i}, \forall i \in \Z  \}.
 \end{align}

\begin{theorem}\label{thm:dismaingeneral}
	Given any family of probability distribution $\{\nm_i\}_{i \in \Z} \in  \cup_{n \ge 1} \KM_{n}$,  there exists a probability distribution $\nm^*$ such that the family of probability distributions  $\{\p^*_i\}_{i \in \Z}$ with $\p^*_i \equiv \p^*$  satisfies the differential privacy constraint \eqref{eqn:disgeneralnoise} and
 	\begin{align}
 		 \sup_{i \in \Z} \sum_{j=-\infty}^{+\infty} \loss(j) \nm^*_i(j)     \le   \sup_{i \in \Z} \sum_{j=-\infty}^{+\infty} \loss(j) \nm_i(j) .
 	\end{align}	
\end{theorem}

\begin{IEEEproof}
	The proof is essentially the same as the proof of Theorem \ref{thm:maingeneral}, and thus is omitted.
\end{IEEEproof}

Theorem \ref{thm:dismaingeneral} states that if we assume the family of noise probability distributions is periodic in terms of $i$ (the period can be arbitrary), then in the optimal mechanism we can assume $\p_i$ does not dependent on $i$. We conjecture that the technical condition can be done away with.


\subsection{Optimal Noise Probability Distribution}

Due to Theorem \ref{thm:dismaingeneral}, we restrict to query-output independent perturbation mechanisms.

Let $q(D)$ be the value of the query function evaluated at dataset $D$. The noise-adding mechanism $\KM$ will output
\begin{align}
 	\KM(D) = q(D) + \noise,
 \end{align}
where $\noise$ is the integer-valued  noise added by the mechanism to the output of query function.
Let $\p$ be the probability distribution of the noise $\noise$. Then the optimization problem we study is
\begin{align}
	\mathop{\text{minimize}}\limits_{ \p} & \ \sum_{i = -\infty}^{+\infty} \loss(i) \p(i) \label{eqn:optdis}\\
	\text{subject to} & \ \p(i) \le e^{\e} \p(i+d), \forall i \in \Z, d \in \Z, |d| \le |\D|.  \label{eqn:dpdiscrete} 
\end{align}

It turns out that when the cost function $\loss(\cdot)$ is symmetric and monotonically increasing for $i \ge 0$, the solution to the above optimization problem is a discrete variant of the staircase mechanism in the continuous setting.

As in the continuous setting, we also assume that the cost function $\loss(\cdot)$ is symmetric and monotonically increasing for $x \ge 0$, i.e.,
\begin{property}\label{propertydis}
	\begin{align}
	\loss(i) &= \loss(-i), \forall i \in \Z \\
	\loss(i) &\le \loss(i), \forall i,j \in \Z, 0 \le i \le j .
\end{align}
\end{property}

The easiest case is $\D = 1$. 
In the case that $\D = 1$, the solution is the geometric mechanism, which was proposed in \cite{Ghosh09}.

Recall $b \triangleq e^{-\e}$.

\begin{theorem}\label{thm:discrete1}
	If the cost function $\loss(\cdot)$ satisfies Property \ref{propertydis} and $\Delta =1$, then the geometric mechanism, which has a probability mass function $\p$ with $\p(i) = \frac{1-b}{1+b}b^{|i|}, \forall i \in \Z$, is the optimal solution to \eqref{eqn:optdis}.
\end{theorem}

\begin{proof}
	See Appendix \ref{sec:discrete_proof}.
\end{proof}

\begin{figure}[t]
\centering
\includegraphics[scale=0.4]{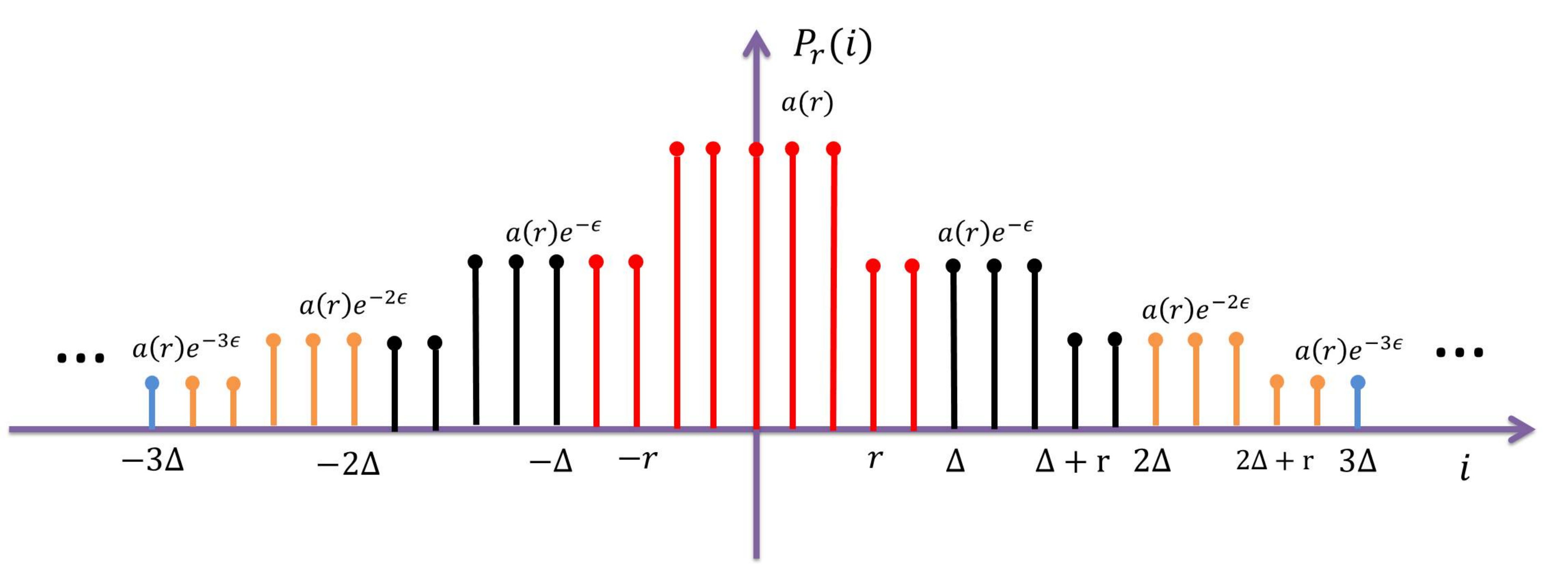}
\caption{The Staircase-Shaped Probability Mass Function $\p_{\mt}(i)$}
\label{fig:disfgamma}
\end{figure}

For fixed general $\D \ge 2$, consider a class of  symmetric and  staircase-shaped probability mass functions defined as follows. Given an integer $ 1 \le \mt \le \D$,   denote $\p_{\mt}$ as the  probability mass function defined by
\begin{align}
\p_{\mt}(i)  =
\begin{cases}
  	 a(\mt) & 0 \le i < \mt  \\
 e^{-\e} a(\mt) & \mt \le i < \D  \\
  e^{-k\e} \p_{\mt}(i - k\D)  &  k \D \le i <   (k+1)\D \; \text{for} \; k \in \N \\
 \p_{\mt}(-i) & i<0
  \end{cases} \label{eqn:defprdis}
\end{align}
for $i \in \Z$, where
\begin{align}
	a(\mt) \triangleq \frac{ 1 - b}{2 \mt + 2 b (\D - \mt) - (1 - b) }.
\end{align}

It is easy to verify that for any $1 \le \mt \le \D$, $\p_{\mt}$ is a valid probability mass function and it satisfies the $\e$-differential privacy constraint \eqref{eqn:dpdiscrete}. We plot the staircase-shaped probability mass function $\p_{\mt}(i)$ in Figure \ref{fig:disfgamma}.

Let $\sP$ be the set of all probability mass functions which satisfy the $\e$-differential privacy constraint \eqref{eqn:dpdiscrete}.

\begin{theorem}\label{thm:discrete2}
For $\D \ge 2$, if the cost function $\loss(x)$ satisfies Property \ref{propertydis}, then
\begin{align}
	\inf_{\p \in \sP}  \sum_{i = -\infty}^{+\infty} \loss(i) \p(i)  = \min_{\{\mt \in \N | 1 \le \mt \le \D \}}  \sum_{i = -\infty}^{+\infty} \loss(i) \p_{\mt}(i).
\end{align}
\end{theorem}

\begin{proof}
	 See Appendix \ref{sec:discrete_proof}.
\end{proof}

Therefore, the optimal noise probability distribution to preserve $\e$-differential privacy for integer-valued query function has a staircase-shaped probability mass function, which is specified by three parameters $\e$, $\D$ and $\mt^* = \mathop{\arg \min} \limits_{\{\mt \in \N | 1 \le \mt \le \D \} } \sum_{i = -\infty}^{+\infty} \loss(i) \p_{\mt}(i)$. In the case $\D = 1$, the staircase-shaped probability mass function is reduced to the geometric mechanism.





\section{Extension to The Abstract Setting} \label{sec:abstractsetting}

In this section, we show how to extend the staircase mechanism to the abstract setting. The approach is essentially the same as the exponential mechanism in \cite{McSherry07}, except that we replace the exponential function by the staircase function.

Consider a privacy mechanism which maps  an input from a domain $\database$ to some output in a range $\mathcal{R}$. Let $\mu$ be the base measure of $\mathcal{R}$. In addition, we have a cost function $\mathcal{C}: \database \times \mathcal{R} \rightarrow [0, +\infty)$.

Define $\D $ as
\begin{align}
 	\D  \triangleq \max_{r \in \mathcal{R}, \;   D_1,D_2 \subseteq \database: |D_1 - D_2| \le 1} | \mathcal{C}(D_1,r) - \mathcal{C}(D_2,r)|,
 \end{align}
i.e., the maximum difference of cost function for any two inputs which differ only on one single value over all $r \in \mathcal{R}$  \cite{McSherry07}.

A randomized mechanism $\KM$ achieves $\e$-differential privacy if for any
$D_1,D_2 \subseteq \database$ such that $|D_1 - D_2| \le 1$, and for any measurable subset $S \subset \mathcal{R}$,
\begin{align}
	\text{Pr}[\KM(D_1) \in S] \le \exp(\e) \;  \text{Pr}[\KM(D_2) \in S].
\end{align}

\begin{definition}[staircase mechanism in the abstract setting] \label{def:stairabstract}
For fixed $\gamma \in [0,1]$, given input $D \in \database$, the staircase mechanism in the abstract setting will output an element in $\mathcal{R}$ with the probability distribution defined as
\begin{align}
 	\p_D (S) = \frac{ \int_{r \in S} f_{\gamma}(\mathcal{C}(D, r))  \mu(dr) }{\int_{r \in \mathcal{R}} f_{\gamma}(\mathcal{C}(D, r))  \mu(dr) },  \forall \; \text{measurable set} \; S \subset \mathcal{R},
 \end{align}
where $f_{\gamma}$ is the staircase-shaped function defined in \eqref{eqn:deffgamma}.
\end{definition}

\begin{theorem}
	The staircase mechanism in the abstract setting in Definition \ref{def:stairabstract} achieves $2\e$-differential privacy.
\end{theorem}

\begin{IEEEproof}
	For any $D_1,D_2 \in \database$ such that $|D_1 - D_2| \le 1$, and for any measurable set $S \subset \mathcal{R}$,
	\begin{align}
		\p_{D_1} (S) &= \frac{ \int_{r \in S} f_{\gamma}(\mathcal{C}(D_1, r))  \mu(dr) }{\int_{r \in \mathcal{R}} f_{\gamma}(\mathcal{C}(D_1, r))  \mu(dr) } \\
		&\le e^{\e} \frac{ \int_{r \in S} f_{\gamma}(\mathcal{C}(D_2, r))  \mu(dr) }{\int_{r \in \mathcal{R}} f_{\gamma}(\mathcal{C}(D_1, r))  \mu(dr) }  \\
		&\le e^{\e} e^{\e} \frac{ \int_{r \in S} f_{\gamma}(\mathcal{C}(D_2, r))  \mu(dr) }{\int_{r \in \mathcal{R}} f_{\gamma}(\mathcal{C}(D_2, r))  \mu(dr) } \\
		&= e^{2\e} \p_{D_2} (S),
	\end{align}
	where we have used the property that $f_{\gamma}(\mathcal{C}(D_1, r)) \le e^{\e} f_{\gamma}(\mathcal{C}(D_2, r))$ and $f_{\gamma}(\mathcal{C}(D_2, r)) \le e^{\e} f_{\gamma}(\mathcal{C}(D_1, r))$ for all $r \in \mathcal{R}$.

	Therefore, the staircase mechanism in the abstract setting achieves $2\e$-differential privacy for any $\gamma \in [0,1]$.

\end{IEEEproof}

In the case that the output range $\mathcal{R}$ is the set of real numbers $\R$ and the cost function $\mathcal{C}(d, r) = |r - q(d)|$ for some real-valued query function $q$,  the above mechanism is reduced to the staircase mechanism in the continuous setting.




\section{Conclusion} \label{sec:conclusion}
In this work we show that adding query-output independent noise with staircase distribution is optimal among all randomized mechanisms (subject to a mild technical condition) that preserve  differential privacy. The optimality is  for  single real-valued query function under a very general utility-maximization (or cost-minimization) framework. The class of optimal noise probability distributions has  {\em staircase-shaped} probability density functions which are symmetric (around the origin), monotonically decreasing and geometrically decaying for $x \ge 0$. The staircase mechanism can be viewed as a {\em geometric mixture of uniform probability distributions}, providing a simple algorithmic description for the mechanism.  Furthermore, the staircase mechanism naturally generalizes to discrete query output settings as well as more abstract settings. 

We explicitly derive the parameter of the staircase mechanism  with minimum expectation of noise amplitude and power.  Comparing the optimal performances with those of the standard Laplacian mechanism, we show that in the high privacy regime ($\e$ is small), Laplacian mechanism is asymptotically optimal as $\e \to 0$;  in the low privacy regime ($\e$ is large), the minimum expectation of noise amplitude  and minimum noise power are $\Theta(\D e^{-\frac{\e}{2}})$
and $\Theta(\D^2 e^{-\frac{2\e}{3}})$ as $\e \to +\infty$, while  the expectation of noise amplitude and power using the Laplacian mechanism are $\frac{\D}{\e}$ and $\frac{2\D^2}{\e^2}$, where $\D$ is the sensitivity of the query function. We conclude that the gains are more pronounced in the moderate to low privacy regime.










\appendices

\section{Proof of Theorem \ref{thm:maingeneral}} \label{app:optimality}

We first give two lemmas on the properties of $\{\nm_{t}\}_{t \in \R}$ which satisfies \eqref{eqn:diffgeneralnoise}.

\begin{lemma}\label{lem:shiftm}
Given 	$\{\nm_{t}\}_{t \in \R}$ satisfying \eqref{eqn:diffgeneralnoise}, and given any scalar $\alpha \in \R$, consider the family of noise probability measures $\{\nm^{(\alpha)}_{t}\}_{t \in \R}$ defined by
\begin{align}
	\nm^{(\alpha)}_{t} \triangleq \nm_{t+\alpha}, \forall t \in \R. \label{eqn:defshift}
\end{align}

Then $\{\nm^{(\alpha)}_{t}\}_{t \in \R}$ also satisfies the differential privacy constraint, i.e.,  $\forall |t_1 - t_2| \le \D$,
\begin{align}
	\nm^{(\alpha)}_{t_1} (S) \le e^{\e} \nm^{(\alpha)}_{t_2}(S + t_1 - t_2).\label{eqn:satequality}
\end{align}
Furthermore, $\{\nm_{t}\}_{t \in \R}$ and $\{\nm^{(\alpha)}_{t}\}_{t \in \R}$ have the same cost, i.e.,
\begin{align}
	\sup_{t \in \R} \int_{x \in \R} \loss(x) \nm_t(dx) = \sup_{t \in \R} \int_{x \in \R} \loss(x) \nm^{(\alpha)}_t(dx). \label{eqn:samecost}
\end{align}
\end{lemma}

\begin{IEEEproof}
	Since by definition the family of probability measures $\{\nm^{(\alpha)}_t\}_{t \in \R}$ is a shifted version of $\{\nm_t\}_{t \in \R}$, \eqref{eqn:samecost} holds.

	Next we show that $\{\nm^{(\alpha)}_t\}_{t \in \R}$ satisfies \eqref{eqn:satequality}. Given any $t_1, t_2$ such that $|t_1 - t_2| \le \D$, then for any measurable set $S \subset \R$, we have
	\begin{align}
		\nm^{(\alpha)}_{t_1} &= \nm_{t_1 + \alpha} (S) \\
			&\le e^{\e} \nm_{t_2 + \alpha} (S + (t_1 + \alpha) - (t_2 + \alpha)) \\
			&= e^{\e} \nm_{t_2 + \alpha} (S + t_1 - t_2) \\
			&= e^{\e} \nm^{(\alpha)}_{t_2} (S + t_1 - t_2).
	\end{align}

	This completes the proof.
\end{IEEEproof}

Next we show that given a collection of families of probability measures each of which satisfies the differential privacy constraint \eqref{eqn:diffgeneralnoise}, we can take a convex combination of them to construct a new family of probability measures satisfying \eqref{eqn:diffgeneralnoise} and the new cost is not worse. More precisely,

\begin{lemma}\label{lem:averagem}
 	Given a collection of finite number of  families of probability measures $\{\nm^{[i]}_t\}_{t \in \R}$ ($i \in \{1,2,3,\dots,n\}$), such that  for each $i$,  $\{\nm^{[i]}_t\}_{t \in \R}$  satisfies \eqref{eqn:diffgeneralnoise} and
 	\begin{align}
 		\sup_{t \in \R} \int_{x \in \R} \loss(x) \nm^{[i]}_t(dx)  = Q, \forall i,
 	\end{align}
 	for some real number $Q$,
 	consider the family of probability measures $\{\tilde{\nu}_t\}_{t \in \R}$ defined by
 	\begin{align}
 		\tilde{\nu}_t \triangleq \sum_{i=1}^n c_i \nm^{[i]}_t, \forall t \in \R,
 	\end{align}
 	i.e., for any measurable set $S \subset \R$,
 	\begin{align}
 		\tilde{\nu}_t (S) = \sum_{i=1}^n c_i \nm^{[i]}_t (S),
 	\end{align}
 	where $c_i \ge 0$, and $\sum_{i=1}^n c_i = 1$.

 	Then $\{\tilde{\nu}_t\}_{t \in \R}$ also satisfies the differential privacy constraint \eqref{eqn:diffgeneralnoise}, and
 	\begin{align}
 		\sup_{t \in \R} \int_{x \in \R} \loss(x) \tilde{\nu}_t(dx) \le Q. \label{eqn:}
 	\end{align}
 \end{lemma}

 \begin{IEEEproof}
 	First we show that $\{\tilde{\nu}_t\}_{t \in \R}$ also satisfies the differential privacy constraint \eqref{eqn:diffgeneralnoise}. Indeed, $\forall \, |t_1 - t_2| \le \D$, $\forall$ measurable set $S \subset \R$,
 	\begin{align}
 		\tilde{\nu}_{t_1} (S) &= \sum_{i=1}^n c_i \nm^{[i]}_{t_1} (S) \\
 						&\le \sum_{i=1}^n c_i e^{\e} \nm^{[i]}_{t_2} (S + t_1 - t_2) \\
 						&= e^{\e} \tilde{\nu}_{t_2} (S + t_1 - t_2).
 	\end{align}

 	Next we show that the cost of $\{\tilde{\nu}_t\}_{t \in \R}$ is no bigger than $Q$. Indeed, for any $t \in \R$,
 	\begin{align}
 		\int_{x \in \R} \loss(x) \tilde{\nu}_t(dx) &= \sum_{i=1}^n c_i \int_{x \in \R} \loss(x) \tilde{\nu}^{[i]}_t(dx) \\
 		& \le \sum_{i=1}^n c_i Q \\
 		& = Q.
 	\end{align}
 	Therefore,
 	\begin{align}
 		\sup_{t \in \R} \int_{x \in \R} \loss(x) \tilde{\nu}_t(dx) \le Q.
 	\end{align}
 \end{IEEEproof}

 Applying Lemma \ref{lem:shiftm} and Lemma \ref{lem:averagem}, we can prove the conjecture under the assumption that the family of probability measures $\{\nm_t\}_{t \in \R}$ is piecewise constant and periodic over $t$.

 \begin{IEEEproof}[Proof of Theorem \ref{thm:maingeneral}]
 	We first prove that for any family of probability measures $\{\nm_t\}_{t \in \R} \in \KM_{T,n}$, there exists a new family of probability measures $\{\tilde{\nm}_t\}_{t \in \R} \in \KM_{T,n}$ such that $\tilde{\nm}_t = \tilde{\nm}$ for all $t \in \R$, i.e., the added noise is independent of query output $t$, and
 	\begin{align}
 	 	\sup_{t \in \R} \int_{x \in \R} \loss(x) \tilde{\nm}_t(dx)  \le \sup_{t \in \R} \int_{x \in \R} \loss(x) \nm_t(dx) .
 	 \end{align}

 	 Indeed, consider the collection of probability measures $\{\nm_t^{(i\frac{T}{n})}\}_{t \in \R}$ for $i \in \{0, 1, 2, \dots, n-1\}$, where $\{\nm_t^{(\alpha)}\}$ is defined in  \eqref{eqn:defshift}. Due to Lemma \ref{lem:shiftm}, for all $i$, $\{\nm_t^{(i\frac{T}{n})}\}_{t \in \R}$ satisfies the differential privacy constraint \eqref{eqn:diffgeneralnoise}, and the cost is the same as the cost of $\{\nm_t\}_{t \in \R}$.

 	 Define
 	 \begin{align}
 	 	\tilde{\nm}_t = \sum_{i=0}^{n-1} \frac{1}{n} \nm_t^{(i\frac{T}{n})}.
 	 \end{align}
 	 Then due to Lemma \ref{lem:averagem}, $\{\tilde{\nm}_t\}_{t \in \R}$ satisfies \eqref{eqn:diffgeneralnoise}, and the cost of is not worse, i.e.,
 	 \begin{align}
 	  	\sup_{t \in \R} \int_{x \in \R} \loss(x) \tilde{\nm}_t(dx)  \le \sup_{t \in \R} \int_{x \in \R} \loss(x) \nm_t(dx).
 	  \end{align}

 	  Furthermore, since $\{\nm_t\}_{t \in \R} \in \KM_{T,n}$, for any $t \in \R$,
 	  \begin{align}
 	  	  \tilde{\nm}_t &= \sum_{i=0}^{n-1} \frac{1}{n} \nm_t^{(i\frac{T}{n})} 
 	  	  = \sum_{i=0}^{n-1} \frac{1}{n} \nm_{i\frac{T}{n}}.
 	  \end{align}
 	  Hence,  $\tilde{\nm}_t$ is independent of $t$.

 	  Therefore, among the collection of probability measures in $	\cup_{T >0} \cup_{n \ge 1} \KM_{T,n}$, to minimize the cost we only need to consider the families of noise probability measures which are independent of the query output $t$. Then due to Theorem \ref{thm:main}, the staircase mechanism is optimal among all query-output-independent noise-adding mechanisms. This completes the proof of Theorem \ref{thm:maingeneral}.
 \end{IEEEproof}

\section{Proof of Theorem \ref{thm:main}}\label{sec:proof}
In this section, we give detailed and rigorous proof of Theorem \ref{thm:main}.

\subsection{Outline of Proof}
 The key idea of the proof is to use a sequence of probability distributions with piecewise constant probability density functions to approximate any probability distribution satisfying the differential privacy constraint \eqref{eqn:dpconstraintfinal}.  The proof consists of 8 steps in total, and in each step we narrow down the set of  probability distributions where the optimal probability distribution should lie in:

 \begin{itemize}
 	\item Step 1 proves that we only need to consider symmetric probability distributions.

 	\item Step 2 and Step 3 prove that  we only need to consider probability distributions which have symmetric piecewise constant probability density functions.

 	\item Step 4 proves that we only need to consider those symmetric piecewise constant probability density functions which are monotonically decreasing for $x \ge 0$.

 	\item Step 5 proves that optimal probability density function should periodically decay.

 	\item Step 6, Step 7 and Step 8 prove that the optimal probability density function over the interval $[0, \D)$ is a step function, and they conclude the proof of Theorem \ref{thm:main}.
 \end{itemize}

\subsection{Step 1}

Define
\begin{align}
 	V^* \triangleq \inf_{\p \in \sP}  \int_{x \in \R} \loss (x)  \p(dx). \label{def:Vstar}
 \end{align}
Our goal is to prove that $V^* =  \inf\limits_{\gamma \in [0,1]}  \int_{x \in \R} \loss (x)  \p_{\gamma}(dx) $.

If $V^* = + \infty$, then due to the definition of $V^*$, we have
\begin{align}
	 \inf_{\gamma \in [0,1]}  \int_{x \in \R} \loss (x)  \p_{\gamma}(dx) \ge V^* = + \infty,
\end{align}
and thus $\inf_{\gamma \in [0,1]}  \int_{x \in \R} \loss (x) = V^* = + \infty$. So we only need to consider the case $V^* < +\infty$, i.e., $V^*$ is finite. Therefore, in the rest of the proof, we assume $V^*$ is finite.

First we prove that we only need to consider symmetric probability measures.
\begin{lemma}\label{lem:symmetric}
 	Given $\p \in \sP$, define a symmetric probability distribution $\psym$ as
 	\begin{align}
 		\psym(S) \triangleq \frac{\p(S) + \p(-S)}{2}, \forall \ \text{measurable set} \ S   \subseteq \R,\label{eqn:definesym}
 	\end{align}
 	where the set $ -S \triangleq \{ - x \ | \ x \in S   \} $.
 	Then $\psym \in \sP$, i.e., $\psym$ satisfies the differential privacy constraint \eqref{eqn:dpconstraintfinal}, and
 	\begin{align}
 		\int_{x \in \R} \loss (x)  \psym (dx) = \int_{x \in \R} \loss (x)  \p (dx).
 	\end{align}
 \end{lemma}

\begin{IEEEproof}
	It is easy to verify that $\psym$ is a valid probability distribution. Due to the definition  of   $\psym$ in \eqref{eqn:definesym}, we have
	\begin{align}
		\psym(S) = \frac{\p(S) + \p(-S)}{2}  = \psym(-S),
	\end{align}
	for any measurable set $S \subseteq \R$. Thus, $\psym$ is a symmetric probability distribution.

	Next, we show that $\psym$ satisfies \eqref{eqn:dpconstraintfinal}. Indeed, $\forall$ measurable set $S \subseteq \R$ and   $\forall |d| \le \D$,
	\begin{align}
		\psym(S) &= \frac{\p(S) + \p(-S)}{2} \\
		   & \le \frac{e^{\e} \p(S+d)  + e^{\e} \p(-S-d)}{2} \label{eqn:temp1} \\
		   & = \frac{e^{\e} \p(S+d)  + e^{\e} \p(-(S+d))}{2}  \\
		   & = e^{\e} \psym(S+d),
	\end{align}
	where in \eqref{eqn:temp1} we use the facts $\p(S) \le e^{\e} \p(S+d) $ and $\p(-S) \le e^{\e} \p(-S-d) $.

	Lastly, since $\loss(x)$ is symmetric,
	\begin{align}
		\int_{x \in \R} \loss (x)  \p (dx) & = \int_{x \in \R} \frac{\loss (x) + \loss (-x)} {2}  \p (dx) \\
		&= \int_{x \in \R} \loss (x)  \psym (dx).
	\end{align}
\end{IEEEproof}

Therefore,  if we define
\begin{align}
	\sP_{\text{sym}} \triangleq \{ \psym | \p \in \sP  \},
\end{align}
due to Lemma \ref{lem:symmetric},
 \begin{lemma}
 \begin{align}
 	V^* = \inf_{\p \in \sP_{\text{sym}}}  \int_{x \in \R} \loss (x)  \p(dx) .
 \end{align}
 \end{lemma}

\subsection{Step 2}

Next we prove that for any probability distribution $\p$ satisfying differential privacy constraint \eqref{eqn:dpconstraintfinal}, the probability $\text{Pr}(\noise = x) = 0, \forall x \in \R $, and $\p([y,z]) \neq 0$ for any $y < z \in \R$.

\begin{lemma}\label{lem:no_point_mass}
	 $\forall \p \in \sP, \forall x \in \R$,  $\p(\{x\}) = 0$. And, for any $y < z \in \R$, $\p([y,z]) \neq 0$.
\end{lemma}

\begin{IEEEproof}
	Given $\p \in \sP$, suppose $\p(\{x_0\}) = p_0 >0$, for some $x_0 \in \R$.
	Then for any $x \in [x_0, x_0 + \D]$,
	\begin{align}
		\p(\{x\}) \ge  e^{-\e},
	\end{align}
	due to \eqref{eqn:dpconstraintfinal}.

	So $\p(\{x\})$ is strictly lower bounded by a positive constant for uncountable number of $x$, and thus $\p([x_0, x_0 + \D]) =  + \infty$, which contradicts with the fact $\p$ is a probability distribution.

	Therefore,  $\forall \p \in \sP, \forall x \in \R$,  $\p(\{x\}) = 0$.

    Suppose $\p([y,z]) = 0$ for some $y <z \in \R$. Then from \eqref{eqn:dpconstraintfinal} we have for any $|d| \le \D$,
    \begin{align}
        \p([y+d,z+d]) \le e^{\e}  \p([y,z]) = 0,
    \end{align}
    and thus $\p([y+d,z+d]) = 0$. By induction,  for any $k \in \Z$,
    $\p([y+kd, z+kd]) = 0$, which implies that $\p((-\infty, +\infty)) =0 $. Contradiction. So for any $y < z \in \R$, $\p([y,z]) \neq 0$.
\end{IEEEproof}

\subsection{Step 3}
In this subsection, we show that for any $\p \in \sP_{\text{sym}}$ with
\begin{align}
	V(\p) \triangleq \int_{x \in \R} \loss(x)  \p(dx) < + \infty,
\end{align}
we can use a sequence of probability measures $\{ \p_i \in \sP_{\text{sym}} \}_{i \ge 1}$ with symmetric piecewise constant probability density functions to approximate $\p$ with $\lim_{i \to +\infty} V(\p_i) = V(\p)$.

\begin{lemma}\label{lem:approx}
	Given $\p \in \sP_{\text{sym}}$ with $V(\p) < + \infty$, any positive integer $i \in N$, define  $\p_i$  as the probability distribution with a symmetric probability density function $f_i(x)$ defined as
	\begin{align}
		f_i(x) = \begin{cases}
			a_k \triangleq \frac{ \p([k\frac{D}{i}, (k+1)\frac{D}{i} )}{\frac{D}{i}}  & x\in [k\frac{D}{i}, (k+1)\frac{D}{i} ) \;  \text{for} \; k \in \N \\
			f_i(-x) &   x < 0
		\end{cases} 	
	\end{align}

	Then $\p_i \in \sP_{\text{sym}}$ and
	\begin{align}
		\lim_{i \to +\infty} V(\p_i) = V(\p).
	\end{align}
	
\end{lemma}

\begin{IEEEproof}

	First we prove that $\p_i \in \sP_{\text{sym}}$, i.e., $\p_i$ is symmetric and satisfies the differential privacy constraint \eqref{eqn:dpconstraintfinal}.

	By definition $f_i(x)$ is a symmetric and nonnegative function, and
	\begin{align}
		\int_{-\infty}^{+\infty} f_i (x) dx &= 2 \int_{0}^{+\infty} f_i (x) dx  \\
		&= 2 \int_{x \in [0, +\infty)}  \p (dx) \\
		&= 2 \int_{x \in (0, +\infty)}  \p (dx) \label{eqn:limzero} \\
		& = 1,
	\end{align}
where in \eqref{eqn:limzero} we used the fact $\p(\{0\}) = 0$ due to Lemma \ref{lem:no_point_mass}. In addition, due to Lemma \ref{lem:no_point_mass}, $a_k >0, \forall k \in \N$.

	So  $f_i (x)$ is a valid symmetric probability density function, and thus $\p_i$ is a valid symmetric probability distribution.

	Define the density sequence of $\p_i$  as the sequence $\{a_0, a_1, a_2, \dots, a_n, \dots \}$. Since $\p$ satisfies \eqref{eqn:dpconstraintfinal}, it is easy to see that
	\begin{align}
		a_j \le e^{\e} a_{j+k} \; \text{and} \; a_{j+k} \le e^{\e} a_j, \forall j \ge 0, 0 \le k \le i.
	\end{align}

	Therefore, for any $x,y$ such that $|x-y| \le \D$, we have
	\begin{align}
		f_i(x) \le e^{\e} f_i(y) \; \text{and} \; f_i(y) \le e^{\e} f_i(x),
	\end{align}
	which implies that $\p_i$ satisfies \eqref{eqn:dpconstraintfinal}. Hence, $\p_i \in \sP_{\text{sym}}$.

	Next we show that
	\begin{align}
		\lim_{i \to +\infty} V(\p_i) = V(\p).
	\end{align}

	Since $\loss(x)$ satisfies Property \ref{property2}, we can assume there exists a constant $B > 0$ such that
\begin{align}
 	\loss(x+1) \le B \loss(x), \forall x \ge T.
 \end{align}

Given $\delta >0$, since $V(\p)$ is finite, there exists integer $T^* >T$ such that
\begin{align}
	\int_{x \ge T^*} \loss(x)  \p(dx) < \frac{\delta}{B}.
\end{align}

For any integers $i \ge 1 $, $N \ge T^* $,
\begin{align}
	\int_{x \in [N,N+1)} \loss(x) \p_i (dx) & \le \p_i([N,N+1)) \loss(N+1) \\
	&= \p([N,N+1)) \loss(N+1) \\
	&\le \int_{x \in [N,N+1)} B \loss(x) \p  (dx).
\end{align}

Therefore,
\begin{align}
	\int_{x \in [T^*, +\infty)} \loss(x) \p_i (dx) &\le \int_{x \in [T^*, +\infty)} B \loss(x) \p (dx) \\
	&\le B \frac{\delta}{B} \\
	&= \delta.
\end{align}

For $x\in [0, T^*)$, $\loss(x)$ is a bounded function, and thus by the definition of  Riemann$\text{-}$Stieltjes integral, we have
\begin{align}
	\lim_{i \to \infty} \int_{x \in [0, T^*)} \loss(x) \p_i (dx) = \int_{x \in [0, T^*)} \loss(x) \p (dx).
\end{align}

So  there exists a sufficiently large integer  $i^*$ such that for all $i \ge i^*$
\begin{align}
	\left | \int_{x \in [0, T^*)} \loss(x) \p_{i} (dx) - \int_{x \in [0, T^*)} \loss(x) \p (dx) \right |  \le \delta.
\end{align}

Hence, for all $i \ge i^*$
\begin{align}
	&\; |V(\p_i) - V(\p)| \\
	=&\; \left | \int_{x \in \R} \loss(x) \p_{i} (dx) - \int_{x \in \R} \loss(x) \p (dx) \right | \\
	=&\; 2 \left |\int_{x \in [0, T^*)} \loss(x) \p_{i} (dx) - \int_{x \in [0, T^*)} \loss(x) \p(dx) + \int_{x \in [T^*, +\infty)} \loss(x) \p_{i} (dx)  - \int_{x \in [T^*, +\infty)} \loss(x) \p (dx) \right | \\
	\le &\; 2 \left |\int_{x \in [0, T^*)} \loss(x) \p_{i} (dx) - \int_{x \in [0, T^*)} \loss(x) \p(dx) \right |  + 2  \int_{x \in [T^*, +\infty)} \loss(x) \p_{i} (dx)  + 2 \int_{x \in [T^*, +\infty)} \loss(x) \p (dx) \\
	\le &\; 2 (\delta + \delta + \frac{\delta}{B}) \\
	\le &\; (4+ \frac{2}{B}) \delta.
\end{align}

Therefore,
\begin{align}
	\lim_{i \to +\infty}    \int_{x \in \R } \loss(x) \p_i (dx) = \int_{x \in \R } \loss(x) \p (dx).
\end{align}

\end{IEEEproof}

Define $\sPsymi  \triangleq \{\p_i | \p \in \sP_{\text{sym}}\}$ for $ i \ge 1$, i.e., $\sPsymi$ is the set of probability distributions satisfying differential privacy constraint \eqref{eqn:dpconstraintfinal} and having symmetric piecewise constant (over intervals $[k\frac{\D}{i}, (k+1)\frac{\D}{i} ) \; \forall k \in \N$ ) probability density functions.

Due to Lemma \ref{lem:approx},
\begin{lemma}\label{lem:piecewiseconstant}
	\begin{align}
	V^*  = \inf_{\p \in \cup_{i=1}^{\infty} \sPsymi} \int_{x \in \R } \loss(x) \p (dx).
\end{align}
\end{lemma}

Therefore, to characterize $V^*$, we only need to study probability distributions with symmetric and piecewise constant probability density functions.

\subsection{Step 4}

 Next we show that indeed we only need to consider those probability distributions with symmetric piecewise constant probability density functions which are \emph{monotonically decreasing} when $x \ge 0$.

\begin{lemma}\label{lem:monotone}
  Given $\p_a \in \sPsymi$ with symmetric piecewise constant probability density function $f(\cdot)$, let $\{a_0,a_1,\dots,a_n,\dots\}$ be the density sequence of $f(\cdot)$, i.e,
  \begin{align}
  	f(x) = a_k, x \in [k\frac{\D}{i}, (k+1)\frac{\D}{i} ) \; \forall k \in \N.
  \end{align}
   Then we can construct a new probability distribution $\p_b \in \sPsymi$ the probability density function of which is monotonically decreasing when $x \ge 0$, and
  \begin{align}
   	\int_{x \in \R } \loss(x) \p_b  (dx) \le  \int_{x \in \R } \loss(x) \p_a(dx).
   \end{align}
\end{lemma}

\begin{proof}
Since $a_k > 0, \forall k \in \N$, and
\begin{align}
	\sum_{k=0}^{+\infty} a_k \frac{\D}{i} = \frac{1}{2},
\end{align}
we have $\lim_{k \to +\infty} a_k = 0$.

Given the density sequence $\{a_0,a_1,\dots,a_n,\dots \}$, construct a new monotonically decreasing  density sequence $\{b_0,b_1,\dots,b_n,\dots \}$ and a bijective mapping $\pi: \N \to \N$ as follows
\begin{align}
	I_0 &= \argmax_{k \in \N} a_k, \label{eqn:max1}\\
	\pi (0) &= \min_{n \in I_0} n,  \text{i.e., the smallest element in} \;  I_0, \\
	b_0 &= a_{\pi(0)}, \\
& \\
 \forall    m \in \N  & \; \text{and} \; m \ge 1, \\
		I_m &= \argmax_{k \in \N \backslash \{ \pi(j) |  j< m \}} a_k,  \label{eqn:max2} \\
	\pi (m) &= \min_{n \in I_m} n,  \text{i.e., the smallest element in} \; I_m, \\
	b_m &= a_{\pi(m)}.
\end{align}


Since the sequence $\{a_k\}$ converges to 0,  the maximum of $\{a_k\}$ always exists in \eqref{eqn:max1} and \eqref{eqn:max2}. Therefore, $I_m$ is well defined for all $m \in \N$.

Note that since $\sum_{k=0}^{\infty} a_k \frac{\D}{i} = \frac{1}{2}$ and the sequence $\{b_k\}_{k \in \N}$ is simply a permutation of $\{a_k\}_{k \in \N}$, $\sum_{k=1}^{\infty} b_k \frac{\D}{i} = \frac{1}{2}$.

Therefore, if we define a function $g(\cdot)$ as
	\begin{align}
		g(x) = \begin{cases}
			b_k  & x\in [k\frac{D}{i}, (k+1)\frac{D}{i} ) \;  \text{for} \; k \in \N \\
			g(-x) &   x < 0
		\end{cases} 	
	\end{align}
 then $g(\cdot)$  is a valid symmetric probability density function, and
  \begin{align}
   	\int_{x\in \R} \loss(x) g(x) dx \le \int_{x\in \R} \loss(x) f(x) dx.
   \end{align}

Next, we prove that the probability distribution $\p_b$ with probability density function $g(\cdot)$   satisfies the differential privacy constraint \eqref{eqn:dpconstraintfinal}. Since $\{b_k\}_{k \in \N}$ is a monotonically  decreasing sequence, it is sufficient and necessary to prove that  for all $k \in \N$,
\begin{align}
		\frac{b_k}{b_{k+i}} \le e^{\e}.
\end{align}

To simplify notation, given $k$, we define
\begin{align}
	a^*(k) = \min_{ k \le j \le k+i} a_k,
\end{align}
i.e., $a^*(k)$ denotes the smallest number of  $\{ a_{k}, a_{k+1}, \dots, a_{k+i} \}$.

First, when $k = 0$, it is easy to prove that $\frac{b_0}{b_{i}} \le e^{\e}$. Indeed, recall that $b_0 = a_{\pi(0)}$ and consider the $i+1$ consecutive numbers $\{a_{\pi(0)}, a_{\pi(0)+1}, \dots, a_{\pi(0)+i} \}$ in the original sequence $\{a_k\}_{k \in \N}$.    Then  $a^*(0) \le b_i$, since $b_i$ is the $(i+1)$th largest number in the sequence $\{a_k\}_{k \in \N}$. Therefore,
\begin{align}
	\frac{b_0}{b_{i}}  = \frac{a_{\pi(0)}}{b_{i}} \le  \frac{a_{\pi(0)}}{a^*(0)} \le e^{\e}.
\end{align}

For $k = 1$, $b_1 = a_{\pi(1)}$ and consider the $i+1$ consecutive numbers $\{a_{\pi(1)}, a_{\pi(1)+1}, \dots, a_{\pi(1)+i} \}$. If $\pi(0) \notin [\pi(1), \pi(1)+i]$, then $a^*(\pi(1)) \le b_{i+1}$, and thus
\begin{align}
	\frac{b_1}{b_{i+1}}  = \frac{a_{\pi(1)}}{b_{1+i}} \le  \frac{a_{\pi(1)}}{a^*(\pi(1))} \le e^{\e}.
\end{align}
If $\pi(0) \in [\pi(1), \pi(1)+i]$, then $a^*(\pi(0)) \le b_{i+1}$ and $\frac{a_{\pi(0)}}{a^*(\pi(0))} \le e^{\e}$. Therefore,
\begin{align}
	\frac{b_1}{b_{i+1}}  \le \frac{b_0}{b_{1+i}} \le  \frac{b_0}{a^*(\pi(0))} \le e^{\e}.
\end{align}
Hence, $\frac{b_k}{b_{k+i}} \le e^{\e}$ holds for $k = 1$.

In general, given $k$,  we prove $\frac{b_k}{b_{k+i}} \le e^{\e}$ as follows. First, if $\pi_{j} \notin [\pi(k), \pi(k)+i], \forall j <k$, then $a^*{\pi(k)} \le b_{k+i}$, and hence
\begin{align}
	\frac{b_k}{b_{i+k}}  = \frac{a_{\pi(k)}}{b_{i+k}} \le  \frac{a_{\pi(k)}}{a^*(\pi(k))} \le e^{\e}.
\end{align}
If there exists $j < k$ and $\pi_{j} \in [\pi(k)+1, \pi(k)+i]$, we use   Algorithm \ref{algo:1} to compute a number $j^*$ such that $ j^* < k $ and $\pi_{j} \notin [\pi(j^*)+1, \pi(j^*)+i], \forall j<k$.

\begin{algorithm}
\caption{}
\label{algo:1}
\begin{algorithmic}
\State $j^*  \gets k$
\While{ there exists some $j <k $ and $\pi_{j} \in [\pi(j^*)+1, \pi(j^*)+i]$}
	\State $j^*  \gets j$
\EndWhile
\State Output $j^*$
\end{algorithmic}
\end{algorithm}

It is easy to show that the   loop in Algorithm \ref{algo:1} will  terminate after at most $k$ steps.

After finding $j^*$, we have $j^* < k$, and $a^*(\pi(j^*)) \le b_{k+i}$. Therefore
\begin{align}
 	\frac{b_k}{b_{i+k}}  \le \frac{a_{\pi(j^*)}}{b_{i+k}} \le  \frac{a_{\pi(j^*)}}{a^*(\pi(j^*))} \le e^{\e}.
 \end{align}

 So $\frac{b_k}{b_{k+i}} \le e^{\e} $  holds for all $k \in \N$. Therefore, $\p_b \in \sPsymi$.

 This completes the proof of Lemma \ref{lem:monotone}.

\end{proof}

Therefore, if we define
\begin{align}
	\pa \triangleq \{ \p | \p \in \sPsymi, \; \text{and} \; \text{the density sequence of}\; \p \; \text{is  monotonically decreasing} \},
\end{align}
  then
due to Lemma \ref{lem:monotone},
\begin{lemma}\label{lem:monotonelem}
	\begin{align}
	V^* = \inf_{\p \in \cup_{i=1}^{\infty} \pa}   \int_{x \in \R}  \loss(x)  \p(dx) .
\end{align}
\end{lemma}

\subsection{Step 5}

Next we show that among all symmetric piecewise constant probability density functions, we only need to consider those which are periodically decaying.

More precisely, given positive integer $i$,
\begin{align}
	\pb \triangleq \{ \p | \p \in \pa, \; \text{and} \; \p \text{ has   density sequence}\;  \{a_0,a_1,\dots,a_n,\dots,\} \; \text{satisfying}  \frac{a_k}{a_{k+i}} = e^{\e}, \forall k \in \N \},
\end{align}
then
\begin{lemma}\label{lem:pd}
\begin{align}
	V^* = \inf_{ \p  \in \cup_{i=1}^{\infty} \pb}  \int_{x \in \R} \loss(x)  \p(dx).
\end{align}	
\end{lemma}

\begin{IEEEproof}
	Due to Lemma \ref{lem:monotonelem}, we only need to consider probability distributions with symmetric and piecewise constant probability density functions which are monotonically decreasing for $x \ge 0$.
 	
 	We first show that given $\p_a \in \pa$ with density sequence $\{a_0,a_1,\dots,a_n,\dots,\}$, if $\frac{a_{0}}{a_{i}} < e^{\e}$, then we can construct a probability distributions $\p_b \in \pa$ with density sequence $\{b_0,b_1,\dots,b_n,\dots,\}$ such that $\frac{b_{0}}{b_{i}} = e^{\e}$ and
 	\begin{align}
 		V(\p_a) \ge V(\p_b).
 	\end{align}

 	Define a new sequence $\{b_0,b_1,\dots,b_n,\dots\}$ by scaling up $a_0$ and scaling down $\{a_1,a_2,\dots\}$. More precisely, let  $\delta = \frac{i}{2D ( (\frac{i}{2D}-a_0)e^{-\e}\frac{a_0}{a_i} + a_0  )} - 1 >0$, and set
 	\begin{align}
 	 	b_0 &= a_0 (1 + \delta), \\
 	 	b_k &= a_k ( 1 - \delta'), \forall \; k \ge 1,
 	 \end{align}
 	 where $ \delta' \triangleq \frac{a_0 \delta}{ \frac{i}{2D} - a_0} >0$,
 	 and we have chosen $\delta$ such that $\frac{b_0}{b_i} = \frac{a_0}{a_k} \frac{\frac{i}{2D} - a_0}{ \frac{i}{2D(1+\delta)} - a_0} = e^{\e}$.

 	It is easy to see the sequence $\{b_0,b_1,\dots,b_n,\dots,\}$ correspond to a valid probability density function and it also satisfies the differential privacy constraint \eqref{eqn:dpconstraintfinal}, i.e.,
 	\begin{align}
 		\frac{b_k}{b_{k+i}} \le e^{\e}, \forall k \ge 0.
 	\end{align}

 	Let $\p_b$ be the probability distribution with $\{b_0,b_1,\dots,b_n,\dots,\}$ as the density sequence of its probability density function. Next we show  $V(\p_b) \le V(\p_a)$.

 	It is easy to compute $V(\p_a)$, which is
 	\begin{align}
 		V(\p_a) = 2\frac{\D}{i} \left( a_0 \int_{0}^{\frac{\D}{i}} \loss(x) dx  + \sum_{k=1}^{\infty} a_k \int_{k\frac{\D}{i}}^{(k+1)\frac{\D}{i}}    \loss(x) dx \right).
 	\end{align}

 	Similarly, we can compute $V(\p_b)$ by
 	\begin{align}
 		V(\p_b) &= 2\frac{\D}{i} \left( b_0 \int_{0}^{\frac{\D}{i}} \loss(x) dx  + \sum_{k=1}^{\infty} b_k \int_{k\frac{\D}{i}}^{(k+1)\frac{\D}{i}}    \loss(x) dx \right) \\
 		&= V(\p_a) + 2\frac{\D}{i}  \left(a_0 \delta \int_{0}^{\frac{D}{i}} \loss(x) dx  -  \delta' \sum_{k=1}^{\infty} a_k \int_{k\frac{D}{i}}^{(k+1)\frac{D}{i}}    \loss(x) dx \right) \\
 		&= V(\p_a) +  2\frac{\D}{i} \frac{a_0 \delta}{ \frac{i}{2\D} - a_0  } \left( \sum_{k=1}^{\infty} a_k \int_{0}^{\frac{\D}{i}} \loss(x) dx - \sum_{k=1}^{\infty} a_k \int_{k\frac{\D}{i}}^{(k+1)\frac{\D}{i}}   \loss(x) dx \right) \\
 		&= V(\p_a) +  2\frac{\D}{i} \frac{a_0 \delta}{ \frac{i}{2\D} - a_0  }  \sum_{k=1}^{\infty} a_k \left(  \int_{0}^{\frac{\D}{i}} \loss(x) dx -     \int_{k\frac{\D}{i}}^{(k+1)\frac{\D}{i}}   \loss(x) dx     \right)  \\
 		&\le V(\p_a),
 	\end{align}
where in the last step we used the fact that $\left(  \int_{0}^{\frac{\D}{i}} \loss(x) dx -     \int_{k\frac{\D}{i}}^{(k+1)\frac{\D}{i}}   \loss(x) dx     \right) \le 0$, since $\loss(\cdot)$ is a monotonically increasing function for $x \ge 0$.

 	Therefore, for given $i \in \N$, we only need to consider $\p \in \pa$ with	 density sequence $\{a_0,a_1,\dots,a_n,\dots\}$ satisfying   $\frac{a_0}{a_i} = e^{\e}$.

 	Next, we argue that among all probability distributions  $\p \in \pa$ with	 density sequence $\{a_0,a_1,\dots,a_n,\dots,\}$ satisfying   $\frac{a_0}{a_i} = e^{\e}$,  we only need to consider those probability distributions with density sequence also  satisfying $\frac{a_1}{a_{i+1}} = e^{\e}$.

 	Given $\p_a \in \pa$ with	 density sequence $\{a_0,a_1,\dots,a_n,\dots\}$ satisfying   $\frac{a_0}{a_i} = e^{\e}$ and $\frac{a_1}{a_{i+1}} < e^{\e}$,  we can construct a new probability distribution $\p_b \in \pa$ with density sequence $\{b_0,b_1,\dots,b_n,\dots\}$  satisfying
 	\begin{align}
 	 	\frac{b_0}{b_i} &= e^{\e}, \\
 	 	\frac{b_1}{b_{i+1}} &= e^{\e},
 	 \end{align}
	and $V(\p_a) \ge V(\p_b)$.

	First, it is easy to see $a_1$ is strictly  less than $a_0$, since if $a_0 = a_1$, then $\frac{a_1}{a_{i+1}} = \frac{a_0}{a_{i+1}} \ge \frac{a_0}{a_i} = e^{\e}$. Then we construct a new density sequence by increasing $a_1$ and decreasing $a_{i+1}$. More precisely, we define a new sequence $\{b_0,b_1,\dots,b_n,\dots\}$ as
	\begin{align}
		b_k &= a_k, \forall k \neq 1, k \neq i+1, \\
		b_1 &= a_1 + \delta, \\
		b_{i+1} &= a_{i+1} - \delta,
	\end{align}
	where $\delta = \frac{e^{\e}a_{i+1} - a_1}{1 + e^{\e}}$ and thus $\frac{b_1}{b_{i+1}} = e^{\e}$.

	It is easy to verify that $\{b_0,b_1,\dots,b_n,\dots\}$ is a valid probability density sequence and the corresponding probability distribution $\p_b$ satisfies the differential privacy constraint \eqref{eqn:dpconstraintfinal}. Moreover, $V(\p_a) \ge V(\p_b)$. Therefore, we only need to consider  $\p \in \pa$ with density sequences $\{a_0,a_1,\dots,a_n,\dots\}$  satisfying $\frac{a_0}{a_i} = e^{\e}$ and $\frac{a_1}{a_{i+1}} = e^{\e}$.


	Use the same argument, we can show that  we only need to consider  $\p \in \pa$ with density sequences $\{a_0,a_1,\dots,a_n,\dots\}$  satisfying
	\begin{align}
		\frac{a_k}{a_{i+k}} &= e^{\e}, \forall k \ge 0.
	\end{align}

 	Therefore,
 	\begin{align}
		V^* = \inf_{ \p  \in \cup_{i=1}^{\infty} \pb}  \int_{x \in \R} \loss(x)  \p(dx).
	\end{align}

\end{IEEEproof}

Due to Lemma \ref{lem:pd}, we only need to consider  probability distribution with symmetric, monotonically decreasing (for $x \ge 0$), and periodically decaying piecewise  constant probability density function. Because of the properties of symmetry and periodically decaying, for  this class of probability distributions, the probability density function over $\R$ is completely determined by the probability density function over the interval $[0, \D)$.

Next, we study what the optimal probability density function should be over the interval $[0,\D)$. It turns out that the optimal probability density function over the interval $[0,\D)$ is a step function. We use the following three steps to prove this result.

\subsection{Step 6}
\begin{lemma}\label{lem:ratio}
	Consider a probability distribution $\p_a \in \pb$ ($i \ge 2$) with density sequence $\{a_0,a_1,\dots, a_n,\dots\}$, and $\frac{a_0}{a_{i-1}} < e^{\e}$. Then there exists a probability distribution $\p_b \in \pb$ with   density sequence $\{b_0,b_1,\dots, b_n,\dots\}$such that  $\frac{b_0}{b_{i-1}} = e^{\e}$, and
	\begin{align}
		V(\p_b) \le  V(\p_a).
	\end{align}
\end{lemma}

\begin{IEEEproof}
	
	For each $ 0 \le k \le (i-1)$, define
	\begin{align}
		w_k \triangleq \sum_{j=0}^{+\infty} e^{-j\e} \int_{(j + \frac{k}{i})\D  }^{(j + \frac{k+1}{i})\D} \loss (x) d x.  \label{eqn:sumw}
	\end{align}

Since $\loss(cdot)$ satisfies Property \ref{property2} and $V^* < \infty$, it is easy to show that the sum of series in \eqref{eqn:sumw} exists and is finite, and thus $w_k$ is well defined for all $ 0 \le k \le (i-1)$. In addition, it is easy to see
\begin{align}
	w_0 \le w_1 \le w_2 \le \cdots \le w_{i-1},
\end{align}
since $\loss (x)$ is a monotonically increasing function when $x \ge 0 $.

Then
\begin{align}
	V(\p_a) = \int_{x \in \R} \loss (x)   \p_a(dx) = 2 \sum_{k=0}^{i-1} w_k a_k.
\end{align}

Since $\frac{a_0}{a_{i-1}} < e^{\e}$, we can scale $a_0$ up and scale $\{a_1,\dots,a_{i-1}\}$ down to derive a new valid probability density function with smaller cost. More precisely, define a new probability measure $\p_b \in \pb$ with   density sequence $\{b_0,b_1,\dots, b_n,\dots\}$ via
\begin{align}
	b_0 &\triangleq \gamma a_0, \\
	b_k &\triangleq \gamma' a_k, \forall 1 \le k \le i-1,
\end{align}
for some $\gamma >1 $ and $\gamma' <1$ such that
\begin{align}
 	\frac{b_0}{b_{i-1}} = e^{\e}.
 \end{align}

To make   $\{b_0,b_1,\dots, b_n,\dots\}$ be a valid density sequence, i.e., to make the integral of the corresponding probability density function over $\R$ be 1, we have
\begin{align}
	\sum_{k=0}^{i-1} b_k = \sum_{k=0}^{i-1} a_k = \frac{1-e^{-\e}}{2} \frac{i}{\D}.
\end{align}

Define $t \triangleq \frac{1-e^{-\e}}{2} \frac{i}{\D}$, then we have two linear equations on $\gamma$ and $\gamma'$:
\begin{align}
	\gamma a_0 &= e^{\e} \gamma' \label{eqn:g1} \\
	\gamma a_0 + \gamma' (t - a_0) &= t. \label{eqn:g2}
\end{align}

From \eqref{eqn:g1} and  \eqref{eqn:g2}, we can easily get
\begin{align}
	\gamma &= \frac{e^{\e} t a_{i-1}}{a_0 (t-a_0+e^{\e}a_{i-1})}  >1 \\
	\gamma' &= \frac{t}{ t-a_0+e^{\e}a_{i-1}}  <1.
\end{align}

Then we can  verify that the $V(\p_a) \ge V(\p_a)$. Indeed,
\begin{align}
	&\; V(\p_a) - V(\p_b) \\
	=&\; \int_{x \in \R} \loss (x)  \p_a(dx) - \int_{x \in \R} \loss (x)  \p_b(dx)  \\
	= &\; 2 \sum_{k=0}^{i-1} w_k a_k -  2 \sum_{k=0}^{i-1} w_k b_k \\
	= &\; 2 \left((1 - \gamma) w_0 a_0  + (1 -  \gamma') \sum_{k=1}^{i-1} w_k a_k   \right) \\
	\ge &\; 2 \left((1 - \gamma) w_0 a_0  + (1 -  \gamma') \sum_{k=1}^{i-1} w_0 a_k   \right) \\
	= &\; 2 \left((1 - \gamma) w_0 a_0  + (1 -  \gamma')   w_0 (t  - a_0 )  \right) \\
	= &\; 2 w_0 \left(a_0 - \frac{a_{i-1}e^{\e}t}{t - a_0 + e^{\e}a_{i-1}} + (t - a_0)\frac{-a_0 + e^{\e}a_{i-1}}{t - a_0 + e^{\e}a_{i-1}}    \right)  \\
	= &\; 0.
\end{align}

This completes the proof.

\end{IEEEproof}

Therefore, due to Lemma \ref{lem:ratio}, for all  $i \ge 2$, we only need to consider probability distributions $\p \in \pb$ with density sequence  $\{a_0,a_1,\dots, a_n,\dots\}$ satisfying $\frac{a_0}{a_{i-1}} = e^{\e}$.

More precisely, define
\begin{align}
	\pc = \{ \p \in \pb | \p \; \text{has density sequence} \; \{a_0,a_1,\dots, a_n,\dots\} \; \text{satisfying}  \;  \frac{a_0}{a_{i-1}} = e^{\e}  \}.
\end{align}
Then due to Lemma \ref{lem:ratio},
\begin{lemma}\label{lem:ratiolem}
	\begin{align}
		V^* = \inf_{\p \in \cup_{i=3}^{\infty}\pc}  \int_{x \in \R} \loss (x)  \p(dx).
	\end{align}
\end{lemma}

\subsection{Step 7}

Next, we argue that  for each probability distribution  $\p \in \pc$  ($i \ge 3$) with density sequence $\{a_0,a_1,\dots, a_n,\dots\}$, we can assume that there exists an integer $1 \le k \le (i-2)$, such that
\begin{align}
	a_j &= a_0, \forall 0 \le j <k, \label{eqn:property1} \\
	a_j &= a_{i-1}, \forall k< j <i. \label{eqn:property2}
\end{align}

More precisely,
\begin{lemma}\label{lem:binary}
		Consider a probability distribution $\p_a \in \pc$ ($i \ge 3$) with density sequence $\{a_0,a_1,\dots,a_n,\dots\}$. Then there exists a probability distribution $\p_b \in \pc$ with density sequence $\{b_0,b_1,\dots,b_n,\dots\}$ such that there exists an integer $1 \le k \le (i-2)$ with
		\begin{align}
		b_j &= a_0, \forall \;  0 \le j <k, \label{eqn:binary1}\\
		b_j &= a_{i-1}, \forall \; k< j <i, \label{eqn:binary2}
		\end{align}
and
	\begin{align}
		V (\p_b) \le  V (\p_a). \label{eqn:binary3}
	\end{align}
\end{lemma}

\begin{IEEEproof}
	If there exists integer $1 \le k \le (i-2)$ such that
		\begin{align}
		a_j &= a_0, \forall \; 0 \le j <k, \\
		a_j &= a_{i-1}, \forall \; k< j <i,
		\end{align}
then we can set $\p_b = \p_a$.

Otherwise, let $k_1$ be the smallest integer in $\{0,1,2,\dots,i-1\}$ such that
\begin{align}
	a_{k_1} \neq a_0,
\end{align}
and let $k_2$ be the biggest integer in $\{0,1,2,\dots, i-1\}$ such that
\begin{align}
	a_{k_2} \neq a_{i-1}.
\end{align}

It is easy to see that $k_1 \neq k_2$. Then we can increase $a_{k_1}$ and decrease $a_{k_2}$ simultaneously by the same amount to derive a new probability distribution $\p_b \in \pc$ with smaller cost. Indeed, if
\begin{align}
	a_0 - a_{k_1} \leq a_{k_2} - a_{i-1},
\end{align}
then consider a probability distribution $\p_b  \in \pc$ with  density sequence  $\{b_0,b_1,\dots,b_{i-1},\dots\}$ defined as
\begin{align}
	b_j &= a_0, \forall 0 \le j \le k_1, \\
	b_j &= a_j, \forall k_1 < j \le k_2-1, \\
	b_{k_2} &= a_{k_2} - (a_0 - a_{k_1}),   \\
	b_j &= a_j, \forall k_2 < j \le i-1.
\end{align}

We can verify that  $V (\p_a) \ge  V (\p_b)$ via
\begin{align}
	& \; V (\p_a) -  V (\p_b) \\
 	=& \; \int_{x \in \R} \loss (x)  \p_a(dx) - \int_{x \in \R} \loss (x) \p_b(dx) \\
 	= & \; 2 (w_{k_1} b_{k_1} + w_{k_2} b_{k_2} )   - 2 ( w_{k_1} a_{k_1} + w_{k_2} a_{k_2}) \\
 	= & \; 2  w_{k_1} (a_0 - a_{k_1}) + 2 w_{k_2} ( a_{k_2} - (a_0 - a_{k_1}) - a_{k_2}  ) \\
 	= & \; 2  (a_0 - a_{k_1}) ( w_{k_1}  - w_{k_2}) \\
 	\le & \; 0 ,
 \end{align}
 where $w_i$ is defined in \eqref{eqn:sumw}.

 If $a_0 - a_{k_1} \ge a_{k_2} - a_{i-1}$, then accordingly we can construct $\p_b \in \pc$ by setting
\begin{align}
	b_j &= a_0, \forall 0 \le j < k_1, \\
	b_{k_1} &= a_{k_1} + (a_{k_2} - a_{i-1}), \\
	b_j &= a_j, \forall k_1 < j \le k_2-1, \\
	b_{j} &= a_{i-1}, \forall k_2 \le j \le i-1.
\end{align}

And similarly, it is easy to verify that $V (\p_a) \ge  V (\p_b)$.

Therefore, continue in this way, and finally we will obtain a probability distribution $\p_b \in \pc$ with density sequence $\{b_0,b_1,\dots, b_n,\dots\}$ such that \eqref{eqn:binary1}, \eqref{eqn:binary2} and \eqref{eqn:binary3} hold.

This completes the proof.

\end{IEEEproof}

Define
\begin{align}
	\pd = \{ \p \in \pc \ |\ \p \; \text{has density sequence} \; \{a_0,a_1,\dots,a_n,\dots\} \;  \text{satisfying} \eqref{eqn:binary1} \; \text{and} \; \eqref{eqn:binary2} \; \text{for some}\; 1 \le k \le (i-2)  \}.
\end{align}

Then due to Lemma \ref{lem:binary},
\begin{lemma}\label{lem:stepfunc}
	\begin{align}
		V^* = \inf_{\p \in \cup_{i=3}^{\infty}\pd}  \int_{ x \in \R} \loss (x)  \p(dx) .
	\end{align}
\end{lemma}

\subsection{Step 8}

\begin{IEEEproof}[Proof of Theorem \ref{thm:main}]
	Since $\{ \p_{\gamma} | \gamma \in [0,1] \} \subseteq \sP$, we have
	\begin{align}
		V^* = \inf_{\p \in \sP}  \int_{x \in \R} \loss (x)  \p(dx)  \le \inf_{\gamma \in [0,1]}  \int_{x \in \R} \loss (x)  \p_{\gamma}(dx) .
	\end{align}

	We prove the reverse direction in the following.

	We first prove that for any $\p \in \pd$ ( $i \ge 3$), there exists $\gamma \in [0,1]$ such that
	\begin{align}
		\int_{x \in \R} \loss (x)  \p_{\gamma}(dx) \le \int_{x \in \R} \loss (x)  \p(dx).
	\end{align}

	Consider the density sequence $\{a_0,a_1,\dots, a_n,\dots\}$ of $\p$. Since $\p \in \pd$, there exists an integer $ 0 \le k \le i-2 $ such that
	\begin{align}
		a_j &= a_0, \forall 0 \le j< k, \\
		a_j &= a_0 e^{-\e}, \forall k < j \le i-1.
	\end{align}

	Let
	\begin{align}
		\gammap \triangleq \frac{\frac{1-e^{-\e}}{2\D} - a_0 e^{-\e}}{a_0(1-e^{-\e})} \in [0,1].
	\end{align}
	Then $a(\gammap) = a_0$.

	It is easy to verify that  
	\begin{align}
		k \frac{\D}{i} \le \gammap \D \le (k+1) \frac{\D}{i}.
	\end{align}

	The probability density functions of $\p$ and $\p_{\gammap}$ are the same when $ x \in [0, \frac{k}{i}\D) \cup [\frac{k+1}{i}\D, \D)$. Since the integral of probability density functions over $[0, \D)$ is  $\frac{1 - e^{-\e}}{2}$ due to the periodically decaying property, we have
	\begin{align}
		a_k \frac{\D}{i} = a_0 (\gammap - \frac{k}{i}) \D + e^{-\e}a_0 ( \frac{k+1}{i} - \gammap)\D.
	\end{align}

	Define $\beta \triangleq i (  \gammap - \frac{k}{i} ) \in [0,1]$. Then
	\begin{align}
		a_k = \beta a_0 + (1 - \beta) e^{-\e}a_0.
	\end{align}

    Define
    \begin{align}
		w_k^{(1)} &\triangleq \sum_{j=0}^{+\infty} e^{-j\e} \int_{(j + \frac{k}{i})\D  }^{(j + \gammap) \D} \loss (x) d x,  \label{eqn:sumw1} \\
    w_k^{(2)} &\triangleq \sum_{j=0}^{+\infty} e^{-j\e} \int_{(j +  \gammap) \D    }^{(j + \frac{k+1}{i} )\D} \loss (x) d x,  \label{eqn:sumw2}.
	\end{align}
    Note that $w_k = w_k^{(1)} + w_k^{(2)}$. Since $\loss(x)$ is a monotonically increasing function when $x \ge 0$, we have
    \begin{align}
    	\frac{w_k^{(2)}}{w_k^{(1)}} \ge \frac{(j + \frac{k+1}{i} )\D -  (j +  \gammap) \D  }{(j + \gammap) \D - (j + \frac{k}{i})\D} = \frac{\frac{k+1}{i} - \gammap}{\gammap - \frac{k}{i} }.
    \end{align}

	Therefore,
	\begin{align}
 	& \int_{x \in \R} \loss (x)   \p(dx) - \int_{x \in \R} \loss (x)  \p_{\gammap}(dx) \\
 	= & 2 w_{k} a_{k}  - 2 \left( w_k^{(1)}  a_0 +  w_k^{(2)}  a_0 e^{-\e}\right) \\
 	= & 2  \left(w_k^{(1)} + w_k^{(2)}\right) a_k - 2 \left( w_k^{(1)}  a_0 +  w_k^{(2)}  a_0 e^{-\e}\right)  \\
 	= &  2  ( a_k - a_0 e^{-\e} )w_k^{(2)} -  2  (a_0 - a_k ) w_k^{(1)} .
 	\end{align}

 	Since
 	\begin{align}
 		\frac{a_k - a_0 e^{-\e} }{a_0 - a_k} &= \frac{ \beta (a_0 - a_0 e^{-\e})}{(1-\beta) (a_0 - a_0 e^{-\e}) } \\
 		&= \frac{\beta}{1 - \beta} \\
 		&= \frac{\gammap - \frac{k}{i}}{\frac{k+1}{i} - \gammap} \\
 		&\ge \frac{w_k^{(1)}}{w_k^{(2)}},
 	\end{align}
 	we have
 	\begin{align}
 		& \int_{x \in \R} \loss (x)   \p(dx) - \int_{x \in \R} \loss (x)  \p_{\gammap}(dx) \\
 		= &  2  ( a_k - a_0 e^{-\e} )w_k^{(2)} -  2  (a_0 - a_k ) w_k^{(1)} \\
 		\ge & 0.
 	\end{align}

 	Therefore,
 	\begin{align}
 		V^* &= \inf_{\p \in \cup_{i=3}^{\infty}\pd}  \int_{x \in \R} \loss (x)   \p(dx) \\
 		& \ge \inf_{\gamma \in [0,1]}  \int_{x \in \R} \loss (x)   \p_{\gamma}(dx) .
 	\end{align}

 	We conclude
 	\begin{align}
 		V^* = \inf_{\p \in \sP}  \int_{x \in \R} \loss (x)  \p(dx)  = \inf_{\gamma \in [0,1]}  \int_{x \in \R} \loss (x)  \p_{\gamma}(dx) = \inf_{\gamma \in [0,1]}  \int_{x \in \R} \loss (x)  f_{\gamma}(x) dx.
 	\end{align}

 	This completes the proof of Theorem \ref{thm:main}.
\end{IEEEproof}

\section{Proof of Theorem \ref{thm:1}}\label{app:1}

\begin{IEEEproof}[Proof of Theorem \ref{thm:1}]

Recall $b \triangleq e^{-\e}$, and $\loss(x) = |x|$.  We can compute  $V(\p_\gamma)$ via
\begin{align}
	V(\p_\gamma) &= \int_{x \in \R} |x| f^{\gamma}(x) dx \\
	&= 2 \int_{0}^{+\infty} x f^{\gamma}(x) dx \\
	&= 2 \sum_{k=0}^{+\infty} \left( \int_{0}^{\gamma \D} (x + k\D) a(\gamma) e^{-k\e} dx +  \int_{\gamma \D}^{ \D} (x + k\D) a(\gamma) e^{-\e} e^{-k\e} dx  \right) \\
	&= 2 \D^2 a(\gamma) \sum_{k=0}^{+\infty} \left( e^{-k \e} \frac{(k+\gamma)^2 - k^2}{2} +  e^{-(k+1)\e}  \frac{(k+1)^2 - (k+\gamma)^2}{2} \right) \\
	&= 2 \D^2 a(\gamma) \sum_{k=0}^{+\infty} \left( e^{-k \e} \frac{\gamma^2 + 2 k \gamma}{2} +  e^{-(k+1)\e}  \frac{2k + 1 - 2k\gamma - \gamma^2}{2} \right) \\
	&= 2 \D^2 a(\gamma) \sum_{k=0}^{+\infty} \left( (b + (1-b)\gamma) k
	e^{-k \e}   +  \frac{ b + (1-b)\gamma^2 }{2}  e^{-k \e}    \right) \\
	&= 2 \D^2 a(\gamma) \left(     (b + (1-b)\gamma) \frac{b}{(1-b)^2}   +  \frac{ b + (1-b)\gamma^2 }{2}  \frac{1}{1-b}   \right) \label{eqn:app1}\\
	&= 2 \D^2 \frac{1-b}{2\D (b + (1-b)\gamma)} \left(     (b + (1-b)\gamma) \frac{b}{(1-b)^2}   +  \frac{ b + (1-b)\gamma^2 }{2}  \frac{1}{1-b}   \right)\\
	&= \D \left(  \frac{b}{1-b} + \frac{1}{2}\frac{b + (1-b)\gamma^2}{b + (1-b)\gamma}   \right),
\end{align}
where in \eqref{eqn:app1} we use the formulas
\begin{align}
	\sum_{k=1}^{+\infty} b^k &= \frac{1}{1-b}, \label{eqn:app3}\\
	\sum_{k=1}^{+\infty} k b^k &= \frac{b}{(1-b)^2} \label{eqn:app4}.
\end{align}

Note that  the first term $\frac{b}{1-b}$ is independent of $\gamma$. Define
\begin{align}
 	g(\gamma) \triangleq \frac{b + (1-b)\gamma^2}{b + (1-b)\gamma},
 \end{align}
and thus to minimize $V(\p_\gamma)$ over $\gamma \in [0,1]$, we only need to minimize $g(\gamma)$ over $\gamma \in [0,1]$.

Since $\gamma \in [0,1]$, $g(\gamma) \le 1$. Also note that $g(0) = g(1) = 1$.
So the optimal $\gamma^*$ which minimize $g(\gamma)$ lies in $(0,1)$.

Compute the derivative of $g(\gamma)$ via
\begin{align}
	g'(\gamma) &= \frac{2\gamma(1-b)(b+(1-b)\gamma) - (b+(1-b)\gamma^2)(1-b)}{(b + (1-b)\gamma)^2} \\
	&= (1-b)\frac{(1-b)\gamma^2 + 2b\gamma - b}{(b + (1-b)\gamma)^2}.
\end{align}

Set $g'(\gamma^*) = 0$ and we get
\begin{align}
	\gamma* &= \frac{\sqrt{b}-b}{1-b} \\
	&= \frac{e^{-\frac{1}{2}\e}- e^{-\e}}{1-e^{-\e}} \\
	&= \frac{1}{1 + e^{\frac{\e}{2}}}.
\end{align}

Therefore,
\begin{align}
	V(\p_{\gamma^*}) &= \D \left(  \frac{b}{1-b} + \frac{1}{2}\frac{b + (1-b){\gamma^*}^2}{b + (1-b)\gamma^*}   \right) \\
	&= \D\frac{e^{\frac{\e}{2}}}{e^{\e}-1}.
\end{align}

Due to Theorem \ref{thm:main}, the minimum expectation of noise  amplitude is
$V(\p_{\gamma^*}) = \D\frac{e^{\frac{\e}{2}}}{e^{\e}-1}$.
	
\end{IEEEproof}

\section{Proof of Theorem \ref{thm:3}}\label{app:3}
\begin{IEEEproof}[Proof of Theorem \ref{thm:3}]

Recall $b \triangleq e^{-\e}$. Then we compute $V(\p_\gamma)$ for the cost function $\loss(x) = x^2$ via
\begin{align}
	V(\p_\gamma) &= \int_{x \in \R} x^2 f^{\gamma}(x) dx \\
	&= 2 \int_{0}^{+\infty} x^2 f^{\gamma}(x) dx \\
	&= 2 \sum_{k=0}^{+\infty} \left( \int_{0}^{\gamma \D} (x + k\D)^2 a(\gamma) e^{-k\e} dx +  \int_{\gamma \D}^{ \D} (x + k\D)^2 a(\gamma) e^{-\e} e^{-k\e} dx  \right) \\
	&= 2 \D^3 a(\gamma) \sum_{k=0}^{+\infty} \left( e^{-k \e} \frac{(k+\gamma)^3 - k^3}{3} +  e^{-(k+1)\e}  \frac{(k+1)^3 - (k+\gamma)^3}{3} \right) \\
	&= 2 \D^3 a(\gamma) \sum_{k=0}^{+\infty} \left( e^{-k \e} \frac{\gamma^3 + 3 k \gamma^2 + 3 k^2 \gamma}{3} +  e^{-(k+1)\e}  \frac{3k^2 + 3k + 1 - 3k^2 \gamma - 3k\gamma^2 - \gamma^3}{3} \right) \\
	&= 2 \D^3 a(\gamma) \sum_{k=0}^{+\infty} \left( (\frac{1-\gamma^3}{3} b + \frac{\gamma^3}{3})
	e^{-k \e}   +  (\gamma^2 + (1-\gamma^2)b ) k e^{-k\e} + ( \gamma + (1-\gamma)b ) k^2  e^{-k \e}    \right) \\
	&= 2 \D^3 a(\gamma) \left( (\frac{1-\gamma^3}{3} b + \frac{\gamma^3}{3})
	\frac{1}{1-b}   +  (\gamma^2 + (1-\gamma^2)b ) \frac{b}{(1-b)^2} + ( \gamma + (1-\gamma)b ) \frac{b^2+b}{(1-b)^3}   \right) \label{eqn:app2} \\
	&= 2 \D^3 \frac{1-b}{2\D (b + (1-b)\gamma)}  \left( (\frac{1-\gamma^3}{3} b + \frac{\gamma^3}{3})
	\frac{1}{1-b}   +  (\gamma^2 + (1-\gamma^2)b ) \frac{b}{(1-b)^2} + ( \gamma + (1-\gamma)b ) \frac{b^2+b}{(1-b)^3}   \right) \\
	&= \D^2 \left(  \frac{b^2+b}{(1-b)^2}  + \frac{b + (1-b)\gamma^2}{b + (1-b)\gamma} \frac{b}{1-b} + \frac{1}{3}\frac{b + (1-b)\gamma^3}{b + (1-b)\gamma} \right), \label{eqn:Vx2}
\end{align}
where in \eqref{eqn:app2} we use formulas \eqref{eqn:app3}, \eqref{eqn:app4} and
\begin{align}
	\sum_{k=1}^{+\infty} k^2 b^k &= \frac{(b^2+b)}{(1-b)^3} \label{eqn:app5}.
\end{align}

Note that  the first term $\frac{b^2+b}{(1-b)^2}$ is independent of $\gamma$. Define
\begin{align}
 	h(\gamma) &\triangleq \frac{b + (1-b)\gamma^2}{b + (1-b)\gamma} \frac{b}{1-b} + \frac{1}{3}\frac{b + (1-b)\gamma^3}{b + (1-b)\gamma} \\
 	&= \frac{ \frac{(1-b)\gamma^3}{3} +  b\gamma^2 + \frac{b^2}{1-b}  + \frac{b}{3}  }{ b + (1-b)\gamma },
 \end{align}
and thus to minimize $V(\p_\gamma)$ over $\gamma \in [0,1]$, we only need to minimize $h(\gamma)$ over $\gamma \in [0,1]$.

Since $\gamma \in [0,1]$, $h(\gamma) \le \frac{b}{1-b} + \frac{1}{3}$. Also note that $h(0) = h(1) = \frac{b}{1-b} + \frac{1}{3}$.
So the optimal $\gamma^*$ which minimize $h(\gamma)$ lies in $(0,1)$.

Compute the derivative of $h(\gamma)$ via
\begin{align}
	h'(\gamma) &= \frac{((1-b)\gamma^2 + 2b\gamma)(b+(1-b)\gamma) - (\frac{1-b}{3}\gamma^3 + b \gamma^2 + \frac{b^2}{1-b} + \frac{b}{3} )(1-b) }{(b + (1-b)\gamma)^2} \\
	&= \frac{ \frac{2}{3} (1-b)^2 \gamma^3 + 2b(1-b)\gamma^2 + 2b^2 \gamma - \frac{2b^2 + b}{3}	 }{(b + (1-b)\gamma)^2} \\.
\end{align}

Set $h'(\gamma^*) = 0$ and we get
\begin{align}
	\frac{2}{3} (1-b)^2 {\gamma^*}^3 + 2b(1-b){\gamma^*}^2 + 2b^2 {\gamma^*} - \frac{2b^2 + b}{3} = 0. \label{eqn:app6}
\end{align}

Therefore, the optimal $\gamma^*$ is the real-valued root of the cubic equation \eqref{eqn:app6}, which is
\begin{align}
	\gamma^* = -\frac{b}{1-b} + \frac{(b-2b^2+2b^4-b^5)^{1/3}}{2^{1/3} (1-b)^2}. \label{eqn:optimalr}
\end{align}

We plot $\gamma^*$ as a function of $b$ in Figure \ref{fig:optimalr}, and we can see $\gamma^* \to \frac{1}{2}$ as $\e \to 0$, and $\gamma^* \to 0$ as $\e \to +\infty$. This also holds in the case $\loss(x) = |x|$.

Plug \eqref{eqn:optimalr} into \eqref{eqn:Vx2}, and we get the minimum noise power
\begin{align}
	V(\p_{\gamma^*}) 	&= \D^2 \left(  \frac{b^2+b}{(1-b)^2}  + \frac{b + (1-b){\gamma^*}^2}{b + (1-b){\gamma^*}} \frac{b}{1-b} + \frac{1}{3}\frac{b + (1-b){\gamma^*}^3}{b + (1-b){\gamma^*}} \right) \\
	 &= \D^2\frac{ 2^{-2/3} b^{2/3} (1+b)^{2/3} + b }{ (1-b)^2 }.
\end{align}

Due to Theorem \ref{thm:main}, the minimum expectation of noise power is
$V(\p_{\gamma^*}) = \D^2 \frac{ 2^{-2/3} b^{2/3} (1+b)^{2/3} + b }{ (1-b)^2 }$.


\end{IEEEproof}

\section{Proof of Theorem \ref{thm:gammaprop}}\label{app:gammaprop}
\begin{IEEEproof}[Proof of Theorem \ref{thm:gammaprop}]
	Let $n = m+1$, and  define
\begin{align}
	c_i \triangleq \sum_{k=0}^{+\infty} b^k k^i,
\end{align}
for nonnegative integer $i$.

First we compute $V(\p_\gamma)$ via
	\begin{align}
		V(\p_\gamma) &= 2 \sum_{k=0}^{+\infty} \left ( \int_0^{\gamma \D} (x+k\D)^m a(\gamma) e^{-k\e} dx + \int_{\gamma \D}^{\D} (x+k\D)^m a(\gamma) e^{-(k+1)\e} dx    \right ) \\
		&= 2a(\gamma)\D^{m+1}  \sum_{k=0}^{+\infty} \left ( b^k \frac{(k+\gamma)^{m+1} -k^{m+1} }{m+1}  + b^{k+1} \frac{(k+1)^{m+1} - (k+\gamma)^{m+1}}{m+1}  \right ) \\
		&= 2\D^n a(\gamma) \sum_{k=0}^{+\infty}  \left (b^k \frac{ \sum_{i=1}^n   \binom{n}{i} \gamma^i k^{n-i}      }{n} + b b^k \frac{\sum_{i=1}^n \binom{n}{i}(1-\gamma^i)k^{n-i} }{n}  \right ) \\
		&= 2\D^n a(\gamma) \left (\sum_{i=1}^n \frac{\binom{n}{i}\gamma^i c_{n-i}}{n} + b \sum_{i=1}^n \frac{\binom{n}{i}(1-\gamma^i)c_{n-i}}{n}  \right ) \\
		&= 2\D^n a(\gamma) \sum_{i=1}^n \frac{ \binom{n}{i}c_{n-i}(\gamma^i (1-b)+b)}{n} \\
		&= \frac{2\D^n (1-b)}{2\D n} \frac{ \sum_{i=1}^n  \binom{n}{i}c_{n-i}(\gamma^i (1-b)+b) }{\gamma(1-b) + b}.
	\end{align}

Let $h_i(\gamma) \triangleq \frac{ \gamma^i (1-b)+b }{\gamma(1-b) + b}$ for $i \ge 2$. Since $h_i(0) = h_i(1) = 1$ and $h_i(\gamma) < 1 $ for $\gamma \in (0,1)$, $h_i(\gamma)$ achieves the minimum value in the open interval $(0,1)$.

Therefore, if we define $h(\gamma) \triangleq \frac{ \sum_{i=1}^n  \binom{n}{i}c_{n-i}(\gamma^i (1-b)+b) }{\gamma(1-b) + b}$, the optimal $\gamma^* \in [0,1]$, which minimizes $V(\p_\gamma)$, should satisfy
\begin{align}
	h'(\gamma^*) = 0,
\end{align}
where $h'(\cdot)$ denotes the first order derivative of $h(\cdot)$.

It is straightforward to derive the expression for $h'(\cdot)$:
\begin{align}
	h'(\gamma) &= \frac{  (\sum_{i=1}^n  \binom{n}{i}c_{n-i} i \gamma^{i-1}(1-b) )  (\gamma(1-b) + b) - (1-b)  \sum_{i=1}^n  \binom{n}{i}c_{n-i}(\gamma^i (1-b)+b)    } { (\gamma(1-b) + b)^2 } \\
	&= \frac{ \sum_{i=1}^n  \binom{n}{i}c_{n-i} i \gamma^{i}(1-b)^2  + \sum_{i=1}^n  \binom{n}{i}c_{n-i} i \gamma^{i-1}(1-b)b - \sum_{i=1}^n  \binom{n}{i}c_{n-i} \gamma^{i}(1-b)^2 - \sum_{i=1}^n  \binom{n}{i}c_{n-i} b(1-b)   }{ (\gamma(1-b) + b)^2 } \\
	&= \frac{ \sum_{i=1}^n  \binom{n}{i}c_{n-i} (i-1) \gamma^{i}(1-b)^2  + \sum_{i=1}^n  \binom{n}{i}c_{n-i} i \gamma^{i-1}(1-b)b  - \sum_{i=1}^n  \binom{n}{i}c_{n-i} b(1-b)   }{ (\gamma(1-b) + b)^2 }. \label{eqn:gammanum}
\end{align}

Therefore, $\gamma^*$ should make the numerator of \eqref{eqn:gammanum} be zero, i.e., $\gamma^*$ satisfies
\begin{align}
	\sum_{i=1}^n  \binom{n}{i}c_{n-i} (i-1) \gamma^{i}(1-b)^2  + \sum_{i=1}^n  \binom{n}{i}c_{n-i} i \gamma^{i-1}(1-b)b  - \sum_{i=1}^n  \binom{n}{i}c_{n-i} b(1-b) = 0.
\end{align}

Since 
\begin{align}
	& \sum_{i=1}^n  \binom{n}{i}c_{n-i} (i-1) \gamma^{i}(1-b)^2  + \sum_{i=1}^n  \binom{n}{i}c_{n-i} i \gamma^{i-1}(1-b)b  - \sum_{i=1}^n  \binom{n}{i}c_{n-i} b(1-b) \\
	= &  \sum_{i=1}^n  \binom{n}{i}c_{n-i} (i-1) \gamma^{i}(1-b)^2 + \sum_{i=0}^{n-1}  \binom{n}{i+1}c_{n-(i+1)} (i+1) \gamma^{i}(1-b)b - \sum_{i=1}^n  \binom{n}{i}c_{n-i} b(1-b) \\
	= & c_0 (n-1)\gamma^n (1-b)^2 + \sum_{i=1}^{n-1} \left(    \binom{n}{i}c_{n-i} (i-1) (1-b)^2   +  \binom{n}{i+1}c_{n-(i+1)} (i+1)(1-b)b  \right) \gamma^{i}  \nonumber \\ & + n c_{n-1}(1-b)b - \sum_{i=1}^n  \binom{n}{i}c_{n-i} b(1-b) \\
	= &  c_0 (n-1)\gamma^n (1-b)^2 + \sum_{i=1}^{n-1} \left(    \binom{n}{i}c_{n-i} (i-1) (1-b)^2   +  \binom{n}{i+1}c_{n-(i+1)} (i+1)(1-b)b  \right) \gamma^{i} - \sum_{i=2}^n  \binom{n}{i}c_{n-i} b(1-b), 
\end{align}
$\gamma^*$ satisfies
\begin{align}
	c_0 (n-1){\gamma^*}^n (1-b)^2 + \sum_{i=1}^{n-1} \left( \binom{n}{i}c_{n-i} (i-1) (1-b)^2   +  \binom{n}{i+1}c_{n-(i+1)} (i+1)(1-b)b  \right) {\gamma^*}^{i}  - \sum_{i=2}^n  \binom{n}{i}c_{n-i} b(1-b) = 0. \label{eqn:gammafinal}
\end{align}

We can derive the asymptotic properties of $\gamma^*$ from \eqref{eqn:gammafinal}. Before deriving the properties of $\gamma^*$, we first study the asymptotic properties of $c_i$, which are functions of $b$.


There are closed-form formulas for $c_i$ (i=0,1,2,3):
\begin{align}
	c_0 &= \sum_{k=0}^{+\infty} b^k  = \frac{1}{1-b}, \\
	c_1 &= \sum_{k=0}^{+\infty} b^k k = \frac{b}{(1-b)^2}, \\
	c_2 &= \sum_{k=0}^{+\infty} b^k k^2 = \frac{b^2+b}{(1-b)^3}, \\
	c_3 &= \sum_{k=0}^{+\infty} b^k k^3 = \frac{b^3+4b^2+b}{(1-b)^4}. \\
\end{align}

In general, for $i \ge 1$,
\begin{align}
 	c_{i+1} &= \sum_{k=0}^{+\infty} b^k k^{i+1} = \sum_{k=1}^{+\infty} b^k k^{i+1} = b + \sum_{k=1}^{+\infty} b^{k+1} (k+1)^{i+1}, \\
 	b c_{i+1} &= \sum_{k=0}^{+\infty} b^{k+1} k^{i+1} = \sum_{k=1}^{+\infty} b^{k+1} k^{i+1}.
 \end{align}

 Therefore,
 \begin{align}
 	c_{i+1} - b c_{i+1} &= b + \sum_{k=1}^{+\infty} b^{k+1} ( (k+1)^{i+1} - k^{i+1} ) \\
 	&= b +  \sum_{k=1}^{+\infty} b^{k+1} \sum_{j=0}^i \binom{i+1}{j} k^j \\
 	&= b +  b  \sum_{j=0}^i \binom{i+1}{j}  \sum_{k=1}^{+\infty} k^j b^k \\
 	&= b + b (\frac{b}{1-b} + \sum_{j=1}^i \binom{i+1}{j} c_j) \\
 	&= \frac{b}{1-b} + b \sum_{j=1}^i \binom{i+1}{j} c_j,
 \end{align}
and thus
\begin{align}
	c_{i+1} = \frac{b}{(1-b)^2} + \frac{b}{1-b} \sum_{j=1}^i \binom{i+1}{j} c_j. \label{eqn:ci}
\end{align}

From \eqref{eqn:ci}, by induction we can easily prove that 
\begin{itemize}
	\item as $b \to 0$, $c_i \to 0, \forall i \ge 1$;
	\item as $b \to 1$, $\forall i \ge 0, c_i \to +\infty, c_i = \Omega (\frac{i!}{(1-b)^{i+1}})$ and 
	\begin{align}
		\lim_{b \to 1} 	\frac{c_{i+1}}{c_i}(1-b) = i+1.
	\end{align}
\end{itemize}


As $b \to 0$, since $c_i \to 0$ for $i \ge 1$ and $c_0 = 1$, the last two terms of \eqref{eqn:gammafinal} go to zero, and thus from \eqref{eqn:gammafinal} we can see that $\gamma^*$ goes to zero as well.

As $b \to 1$, since $c_i = \Omega (\frac{1}{(1-b)^{i+1}})$ and $\gamma^*$ is bounded by 1, the first term of \eqref{eqn:gammafinal} goes to zero, and the dominated terms in \eqref{eqn:gammafinal} are
\begin{align}
	\binom{n}{2}c_{n-2} 2(1-b)b\gamma^*  - \binom{n}{2}c_{n-2}b(1-b) = 0.
\end{align}
Thus, in the limit we have $\gamma^* = \frac{1}{2}$. Therefore, as $b \to 1$, $\gamma^* \to \frac{1}{2}$.

This completes the proof.
\end{IEEEproof}

\section{Proof of Theorem \ref{thm:discrete1} and Theorem \ref{thm:discrete2}}\label{sec:discrete_proof}
In this section, we prove Theorem \ref{thm:discrete1} and Theorem \ref{thm:discrete2}, which give the optimal noise-adding mechanisms in the discrete setting.

\subsection{Outline of Proof}

The proof technique is very similar to the proof in the continuous settings in Appendix \ref{sec:proof}. The proof consists of 5 steps in total, and in each step we narrow down the set of  probability distributions where the optimal probability distribution should lie in:
 \begin{itemize}
 	\item Step 1 proves that we only need to consider probability mass functions which are monotonically increasing for $i \le 0$ and monotonically decreasing for $i \ge 0$.

 	\item Step 2 proves that  we only need to consider symmetric probability mass functions.

 	\item Step 3 proves that we only need to consider  symmetric  probability mass functions which have periodic and geometric decay for $i \ge 0$, and this proves Theorem \ref{thm:discrete1}.

 	\item Step 4 and Step 5  prove that the optimal probability mass function over the interval $[0, \D)$ is a discrete step function, and they conclude the proof of Theorem \ref{thm:discrete2}.
 \end{itemize}

\subsection{Step 1}

Recall  $\sP$ denotes the set of all probability mass functions which satisfy the $\e$-differential privacy constraint \eqref{eqn:dpdiscrete}. Define
\begin{align}
 	V^* \triangleq \inf_{\p \in \sP}  \sum_{i = -\infty}^{+\infty} \loss(i) \p(i). 
 \end{align}

First we prove that we only need to consider probability mass functions which are monotonically increasing for $i \le 0$ and monotonically decreasing for $i \ge 0$.

Define
\begin{align}
	\pe \triangleq \{ \p \in \sP | \p(i) \le \p(j), \p(m) \ge \p(n), \forall i \le j \le 0, 0 \le m \le n \}.
\end{align}

\begin{lemma}\label{lem:doublesizemono}
	 \begin{align}
	 	V^* = \inf_{\p \in \pe}  \sum_{i = -\infty}^{+\infty} \loss(i) \p(i).
	 \end{align}
\end{lemma}

\begin{IEEEproof}
	We will prove that given a probability mass function $\p_a \in \sP$, we can construct a new probability mass function $\p_b \in \pe$ such that
	\begin{align}
		\sum_{i = -\infty}^{+\infty} \loss(i) \p_a(i) \ge \sum_{i = -\infty}^{+\infty} \loss(i) \p_b(i).
	\end{align}

	Given $\p_a \in \sP$, consider the sequence $sa = \{\p_a(0), \p_a(1), \p_a(-1), \p_a(2), \p_a(-2), \dots \}$. Use the same argument in Lemma \ref{lem:no_point_mass} and we can show $\p_a(i) > 0,  \forall \, i \in \Z$.
	Let the sequence $sb = \{b_0,b_1,b_{-1},b_2, b_{-2}, \dots \}$ be a permutation of the sequence $sa$ in descending order. Since $\sum_{i = -\infty}^{+\infty}   \p_a(i) = 1$, $\lim_{i \to -\infty} \p_a(i) = \lim_{i \to +\infty} \p_a(i) = 0$, and thus $sb$ is well defined. Let $\pi$ be the corresponding permutation mapping, i.e., $\pi: \Z \to \Z$, and
	\begin{align}
		b_i = \p_a(\pi(i)).
	\end{align}

	Since $\loss(\cdot)$ is a symmetric function and monotonically decreasing for $i \ge 0 $, we have
	\begin{align}
		\loss(0) \le \loss(1) \le \loss(-1) \le \loss(2) \le \loss(-2) \le \cdots \le \loss(i) \le \loss(-i) \le \loss(i+1) \le \loss(-(i+1)) \le \cdots.
	\end{align}
	Therefore, if we define a probability mass function  $\p_b$ with
	\begin{align}
		\p_b(i) = b_i, \forall i \in \Z,
	\end{align}
	then
	\begin{align}
		\sum_{i = -\infty}^{+\infty} \loss(i) \p_a(i) \ge \sum_{i = -\infty}^{+\infty} \loss(i) \p_b(i).
	\end{align}

	Next, we only need to prove $\p_b \in \pe$, i.e., we need to show that $\p_b$ satisfies the differential privacy constraint \eqref{eqn:dpdiscrete}.

	Due to the way how we construct the sequence $sb$, we have
	\begin{align}
		b_0 \ge b_1 \ge b_2 \ge b_3 \ge \cdots, \\
		b_0 \ge b_{-1} \ge b_{-2} \ge b_{-3} \ge \cdots.
	\end{align}

Therefore, it is both sufficient and necessary to prove that
\begin{align}
	\frac{b_i}{ b_{i+\D}} &\le e^{\e}, \forall i \ge 0, \\
	\frac{b_i}{ b_{i-\D}} &\le e^{\e}, \forall i \le 0.
\end{align}

Since $\p_a \in \sP$, $\forall \, i \in \{\pi(0)-\D, \pi(0)-\D+1, \pi(0)-\D+2, \dots, \pi(0)+\D \}$,
\begin{align}
	\frac{\p_a(\pi(0))}{\p_a(i)} \le e^{\e}.
\end{align}

Therefore, in the sequence $sb$ there exist at least $2\D$ elements which are no smaller than $b_0 e^{-\e}$. Since $b_{-\D}$ and $b_{\D}$ are the $2\D$th and $(2\D-1)$th  largest elements in the sequence $sb$ other than $b_0$, we have $\frac{b_0}{b_{-\D}} \le e^{\e}$ and $\frac{b_0}{b_{\D}} \le e^{\e}$.

In general, given $i \in \Z$, we can use Algorithm \ref{alg:discrete} to find at least $2\D$ elements in the sequence $sb$ which are no bigger than $b_i$ and no smaller than $b_i e^{-\e}$.

More precisely, given $i \in \Z$, let  $j^*_R$ and $j^*_L$ be the output of Algorithm \ref{alg:discrete}. Note that since the while loops in Algorithm \ref{alg:discrete} can take only at most $2(|i|+1)$ steps, the algorithm will always terminate. For all integers $j \in [\pi(j^*_L)-\D, \pi(j^*_L)-1]$, $\p_a(j)$ is no bigger than $b_i$ and is no smaller than $\p_a(j^*_L) e^{-\e}$; and for all integers $j \in [\pi(j^*_R)+1, \pi(j^*_R)+\D]$, $\p_a(j)$ is no bigger than $b_i$ and is no smaller than $\p_a(j^*_R) e^{-\e}$. Since $\p_a(j^*_R), \p_a(j^*_L) \ge b_i$, for all $j \in [\pi(j^*_L)-\D, \pi(j^*_L)-1] \cup [\pi(j^*_R)+1, \pi(j^*_R)+\D]$, $\p_a(j)$ is no bigger than $b_i$ and is no smaller than $b_i e^{-\e}$. Therefore, there exist at least  $2\D$ elements  in the sequence $sb$ which are no bigger than $b_i$ and no smaller than $b_i e^{-\e}$.

If $i \le 0$, then  $b_{i-\D}$ is the $2\D$th largest element in the sequence $sb$ which is no bigger than $b_i$ and no smaller than $b_i e^{-\e}$; and if $i \ge 0$, then $b_{i+\D}$ is the $(2\D-1)$th largest element in the sequence $sb$ which  is no bigger than $b_i$ and no smaller than $b_i e^{-\e}$.  Therefore, we have
\begin{align}
	\frac{b_i}{ b_{i+\D}} &\le e^{\e}, \forall i \ge 0, \\
	\frac{b_i}{ b_{i-\D}} &\le e^{\e}, \forall i \le 0.
\end{align}

This completes the proof of Lemma \ref{lem:doublesizemono}.

\begin{algorithm}
\caption{}
\label{alg:discrete}
\begin{algorithmic}
\State $j^*_R  \gets i$
\While{there exists some $j$ which appears before $i$ in the sequence $\{0,1,-1,2,-2,\dots\}$ and $\pi(j) \in [\pi(j^*_R)+1, \pi(j^*_R)+\D]$}
	\State $j^*_R  \gets j$
\EndWhile
\\
\State $j^*_L  \gets i$
\While{there exists some $j$ which appears before $i$ in the sequence $\{0,1,-1,2,-2,\dots\}$ and $\pi(j) \in [\pi(j^*_L)-\D, \pi(j^*_L)-1]$}
	\State $j^*_L  \gets j$
\EndWhile
\\
\State Output $j^*_R$ and $j^*_L$.
\end{algorithmic}
\end{algorithm}

\end{IEEEproof}

\subsection{Step 2}

Next we prove that we only need to consider symmetric probability mass functions which are monotonically decreasing when $i \ge 0$.

Define
\begin{align}
	\pdissym \triangleq \{ \p \in \pe |\; \p(i) = \p(-i), \forall \, i \in \Z  \}.
\end{align}

\begin{lemma}\label{lem:monosym}
	 \begin{align}
	 	V^* = \inf_{\p \in \pdissym}  \sum_{i = -\infty}^{+\infty} \loss(i) \p(i).
	 \end{align}
\end{lemma}

\begin{IEEEproof}
	The proof is essentially the same as the proof of Lemma \ref{lem:symmetric}.

	Given $\p_a \in \pe$, define a new probability mass function $\p_b$ with
	\begin{align}
		\p_b (i)  \triangleq \frac{\p_a(i) + \p_a(-i) }{2}, \forall i \in \Z.
	\end{align}

	It is easy to see $\p_b$ is a valid probability mass function and symmetric. Since the cost function $\loss(\cdot)$ is symmetric,
	\begin{align}
		\sum_{i = -\infty}^{+\infty} \loss(i) \p_a(i) = \sum_{i = -\infty}^{+\infty} \loss(i) \p_b(i).
	\end{align}

	Next we show that $\p_b$ also satisfies the differential privacy constraint \eqref{eqn:dpdiscrete}. For any $i \in \Z$ and $ |d| \le \D$,
	since $\p_a(i) \le e^{\e} \p_a(i + d)$ and $\p_a(-i) \le e^{\e} \p_a(-i-d)$, we have
	\begin{align}
		\p_b(i) &= \frac{\p_a(i) + \p_a(-i) }{2} \\
		&\le \frac{e^{\e} \p_a(i+d) + e^{\e} \p_a(-i-d) }{2} \\
		&= e^{\e} \p_b(i+d).
	\end{align}
Therefore, $\p_b$ satisfies \eqref{eqn:dpdiscrete}.

Finally, for any $0 \le i \le j$,
\begin{align}
	\p_b(i) &= \frac{\p_a(i) + \p_a(-i) }{2} \\
			&\ge \frac{\p_a(j) + \p_a(-j) }{2} \\
			&= \p_b(j).
\end{align}

So $\p_b \in \pe$, and thus $\p_b \in \pdissym$. We conclude
  \begin{align}
	 	V^* = \inf_{\p \in \pdissym}  \sum_{i = -\infty}^{+\infty} \loss(i) \p(i).
  \end{align}

\end{IEEEproof}

\subsection{Step 3}

Next we show that among all symmetric and monotonically decreasing (for $i \ge 0$) probability mass function, we only need to consider those which are periodically and geometrically decaying.

More precisely, define
\begin{align}
	\pf \triangleq \{  \p \in \pdissym |  \frac{\p(i)}{\p(i+\D)} = e^{\e},  \forall\, i \in \N \}.
\end{align}
Then
\begin{lemma}\label{lem:discretepd}
\begin{align}
	V^* = \inf_{\p \in \pf}  V(\p).
\end{align}	
\end{lemma}

\begin{IEEEproof}
	Due to Lemma \ref{lem:monosym}, we only need to consider probability mass functions which are symmetric and  monotonically decreasing for $i \ge 0$.
 	
 	We first show that given $\p_a \in \pdissym$, if $\frac{\p_a{0}}{\p_a{\D}} < e^{\e}$, then we can construct a probability mass function $\p_b \in \pdissym$  such that $\frac{\p_b{0}}{\p_b{\D}} = e^{\e}$ and
 	\begin{align}
 		V(\p_a) \ge V(\p_b).
 	\end{align}

 	Since $\p_a$ is symmetric,
 	\begin{align}
 		V(\p_a) = \loss(0)\p_a(0) + 2\sum_{i=1}^{+\infty} \loss(i)\p_a(i).
 	\end{align}

 	Suppose $\frac{\p_a{0}}{\p_a{\D}} < e^{\e}$, then define a new symmetric probability mass function $\p_b$ with
 	\begin{align}
 		\p_b(0) &\triangleq (1+\delta) \p_a(0), \\
 		\p_b(i) &\triangleq (1-\delta') \p_a(i), \forall i \in \Z \backslash \{0\},
 	\end{align}
 	where
 	\begin{align}
 		\delta &= \frac{e^{\e}\frac{\p_a(\D)}{\p_a(0)} -1 }{1 + e^{\e}\frac{\p_a(\D)}{1 - \p_a(0)}} > 0 , \\
 		\delta' &= \frac{e^{\e}\frac{\p_a(\D)}{\p_a(0)} -1 }{\frac{1}{\p_a(0)} +  e^{\e}\frac{\p_a(\D)}{\p_a(0)} - 1 } >0 ,
 	\end{align}
 	so that $\frac{\p_b(0)}{\p_b(\D)} = e^{\e}$.

 	It is easy to see $\p_b \in \pdissym$, and
 	\begin{align}
 		&  V(\p_b) - V(\p_a) \\
 		= & \delta \loss(0) \p_a(0)  - 2 \delta'  \sum_{i=1}^{+\infty} \loss(i) \p_a(i) \\
 		\le & \delta \loss(0) \p_a(0)  - 2 \delta' \sum_{i=1}^{+\infty} \loss(0) \p_a(i) \\
 		\le & \delta \loss(0) \p_a(0)  - \delta' \loss(0) (1 - \p_a(0)) \\
 		=& 0 .
 	\end{align}

 	Therefore, we only need to consider $\p \in \pdissym$ satisfying $\frac{\p(0)}{\p(\D)} = e^{\e}$.

 	By using the same argument as in the proof of Lemma \ref{lem:pd}, one can conclude that we only need to consider $\p \in \pdissym$ satisfying
 	\begin{align}
 		\frac{\p(i)}{\p(i+\D)} = e^{\e}, \forall i \in \N.  \label{eqn:period}
 	\end{align}

 	Therefore, 	$V^* = \inf_{\p \in \pf}  V(\p)$.

\end{IEEEproof}

\begin{IEEEproof}[Proof of Theorem \ref{thm:discrete1}]
	In the case that $\D = 1$, due to Lemma \ref{lem:discretepd}, the symmetry property and \eqref{eqn:period} completely characterize the optimal noise probability mass function, which is the geometric mechanism.
\end{IEEEproof}

\subsection{Step 4}

Due to Lemma \ref{lem:discretepd}, the optimal probability mass function $\p$ is completely characterized by $\p(0),\p(1),\dots,\p(\D-1)$. Next we derive the properties of optimal probability mass function in the domain $\{0,1,2,\dots, \D-1\}$.

Since Lemma \ref{lem:discretepd} solves the case $\D = 1$, in the remaining of this section, we assume $\D \ge 2$.

Define
\begin{align}
	\pj_{\lambda} \triangleq \{ \p \in \pf | \exists k \in \{0,1,\dots,\D-2\}, \p(i) = \p(0), \forall i \in \{0,1,\dots,k\}, \p(j) = \lambda \p(0), \forall j \in \{k+1,k+2,\dots,\D-1\} \}.
\end{align}

\begin{lemma}\label{lem:ratiodis}
	\begin{align}
		V^* = \inf_{\p \in \cup_{\lambda \in [e^{-\e}, 1]} \pj_{\lambda} } V(\p).
	\end{align}
\end{lemma}

\begin{IEEEproof}
If $\D = 2$, then for any $\p \in \pf$, we can set $k=0$, and $\p \in \pj_{\frac{\p(\D-1)}{\p(0)}}$. Therefore, Lemma \ref{lem:ratiodis} holds for $\D = 2$.

Assume $\D \ge 3$. First, we prove that we only need to consider probability mass function $\p \in \pf$ such that there exists $k \in \{1,2,\dots, \D-2\}$ with
\begin{align}
	\p(i) &= \p(0), \forall i \in \{0,1,\dots,k-1 \} \label{eqn:ration11}\\
	\p(j) &= \p(\D-1), \forall i \in \{ k+1,k+2,\dots,\D-1 \}. \label{eqn:ration12}
\end{align}

More precisely, let $\p_a \in \pf$, we can construct a probability mass function $\p_b \in \pf$ such that there exists $k$ satisfying \eqref{eqn:ration11} and \eqref{eqn:ration12}, and $V(\p_b) \ge V(\p_a)$.

The proof technique is very similar to proof of Lemma \ref{lem:binary}. Suppose there does not exists such $k$ for $\p_a$, then let $k_1$  be the smallest integer in $\{1,2,\dots,\D-1\}$ such that
\begin{align}
	\p_a(k_1) \neq \p_a(0),
\end{align}
and let $k_2$ be the biggest integer in $\{0,1,\dots,\D-2\}$ such that
\begin{align}
	\p_a(k_2) \neq \p_a(\D-1).
\end{align}

It is easy to see that $k_1 < k_2$, and $k_1 \neq 0$. Then we can increase $\p_a(k_1)$ and decrease $\p_a(k_2)$ simultaneously by the same amount to derive a new probability mass function $\p_b \in \pf$ with smaller cost. Indeed, if
\begin{align}
	\p_a(0) - \p_a(k_1) \leq \p_a(k_2) - \p_a(\D-1),
\end{align}
then consider a probability mass function $\p_b  \in \pf$ with
\begin{align}
	\p_b(i) &= \p_a(0), \forall 0 \le i \le k_1, \\
	\p_b(i) &= \p_a(i), \forall k_1 < i < k_2, \\
	\p_b(k_2) &= \p_a(k_2) - (\p_a(0) - \p_a(k_1)), \\
	\p_b(i) &= \p_a(i), \forall k_2 < i \le \D-1.
\end{align}

Define
\begin{align}
	w_0 &\triangleq \loss(0) + 2\sum_{k=1}^{\infty} \loss(k\D)e^{-k\e},\\
	w_i &\triangleq 2\sum_{k=0}^{\infty} \loss(i+k\D)e^{-k\e}, \forall i \in \{1,2,\dots,\D-1\}.
\end{align}
Note that since $\loss(\cdot)$ is a monotonically decreasing function when $i\ge 0$, we have $w_0 \le w_1 \le \cdots \le w_{\D-1}$.

Then we can verify that  $V (\p_b) \le  V (\p_a)$ via
\begin{align}
	& \; V (\p_b) -  V (\p_a) \\
 	=& \; \sum_{i=0}^{\D-1} \p_b(i)w_i -  \sum_{i=0}^{\D-1} \p_a(i) w_i \\
 	= & \; (\p_a(0) - \p_a(k_1) ) (w_{k_1} - w_{k_2}) \\
 	\le & \; 0 .
 \end{align}

If
\begin{align}
	\p_a(0) - \p_a(k_1) \ge \p_a(k_2) - \p_a(\D-1),
\end{align}
then we can define $\p_b  \in \pf$ by setting
\begin{align}
	\p_b(i) &= \p_a(0), \forall 0 \le i < k_1, \\
	\p_b(k_1) &= \p_a(k_1) + (\p_a(k_2) - \p_a(\D-1)) , \\
	\p_b(i) &= \p_a(i),  \forall k_1 < i < k_2, \\
	\p_b(i) &= \p_a(\D-1), \forall k_2 \le i \le \D-1.
\end{align}

And similarly, we have
\begin{align}
	V (\p_b) -  V (\p_a) &= (\p_a(k_2) - \p_a(\D-1) ) (w_{k_1} - w_{k_2})  \le 0 .
\end{align}

 Therefore, continue in this way, and finally we will obtain a probability mass function  $\p_b \in \pf$ such that there exists $k$ to satisfy \eqref{eqn:ration11} and \eqref{eqn:ration12} and $V (\p_b) \le V (\p_a)$.

From the above argument, we can see that in the optimal solution $\p^* \in \pf$, the probability mass function can only take at most three distinct values for all $i \in \{0,1,\dots,\D-1\}$, which are $\p^*(0),\p^*(k)$ and $\p^*(\D-1)$. Next we show that indeed either $\p^*(k) = \p^*(0)$ and $\p^*(k) = \p^*(\D-1)$, and this will complete the proof of Lemma \ref{lem:ratiodis}.

The optimal probability mass function $\p \in \pf$ can be specified by three parameters $\p(0), \lambda \in [e^{-\e},1]$, $k \in \{1,2,\dots,\D-2\}$ and $\p(k)$. We will show that when $k$ and $\lambda$ are fixed, to minimize the cost, we have either $\p(k) = \p(0)$ or $\p(k) = \p(\D-1) = \lambda \p(0)$.

Since $\sum_{i=-\infty}^{+\infty} \p(i) = 1$,
\begin{align}
	2\frac{ k\p(0) + \p(k) + (\D-k-1)\lambda \p(0) }{ 1 - b} - \p(0) = 1,
\end{align}
and thus $\p(k) = \frac{(1+\p(0))(1-b) - 2\p(0)k - 2\lambda \p(0)(\D-k-1)}{2}$.

The cost for $\p$ is
\begin{align}
	V(\p) &= \p(0) \sum_{i=0}^{k-1}w_i + \p(\D-1) \sum_{i=k+1}^{\D-1}w_i + \p(k)w_k \\
	&= \p(0) \sum_{i=0}^{k-1}w_i + \lambda \p(0) \sum_{i=k+1}^{\D-1}w_i + (\frac{(1+\p(0))(1-b) - 2\p(0)k - 2\lambda \p(0)(\D-k-1)}{2})w_k,
\end{align}
which is a linear function of the parameter $\p(0)$.

Since $\p(k) \ge \lambda \p(0)$ and $\p(k) \le \p(0)$,
we have
\begin{align}
	2\frac{ k\p(0) + \p(k) + (\D-k-1)\lambda \p(0) }{ 1 - b} - \p(0) = 1 \le 2\frac{ k\p(0) + \p(0) + (\D-k-1)\lambda \p(0) }{ 1 - b} - \p(0),\\
	2\frac{ k\p(0) + \p(k) + (\D-k-1)\lambda \p(0) }{ 1 - b} - \p(0) = 1 \ge 2\frac{ k\p(0) + \lambda \p(0) + (\D-k-1)\lambda \p(0) }{ 1 - b} - \p(0),
\end{align}
and thus the constraints on $\p(0)$ are
\begin{align}
	\frac{1-b}{2k+2+2\lambda(\D-k-1) - 1 + b} \le \p(0) \le \frac{1-b}{2k+2\lambda(\D-k) - 1 + b}. \label{eqn:inequaona}
\end{align}

Since $V(\p)$ is a linear function of $\p(0)$, to minimize the cost $V(\p)$, either $\p(0) = \frac{1-b}{2k+2+2\lambda(\D-k-1) - 1 + b} $  or $\p(0) = \frac{1-b}{2k+2\lambda(\D-k) - 1 + b}$, i.e., $\p(0)$ should take one of the two extreme points of \eqref{eqn:inequaona}. To get these two extreme points, we have either $\p(k) = \p(0)$ or $\p(k) = \lambda \p(0) = \p(\D-1)$.

Therefore, in the optimal probability mass function $\p \in \pf$, there exists $k \in \{0,1,\dots,\D-2\}$ such that
\begin{align}
	\p(i) &= \p(0), \forall i \in \{0,1,\dots,k\}\\
	\p(i) &= \p(\D-1), \forall i \in \{k+1,k+2,\dots,\D-1\}.
\end{align}

This completes the proof of Lemma \ref{lem:ratiodis}.

\end{IEEEproof}

\subsection{Step 5}

In the last step, we prove that although $\lambda \in [e^{-\e},1]$, in the optimal probability mass function, $\lambda$ is either $e^{-\e}$ or $1$, and this will complete the proof of Theorem \ref{thm:discrete2}.

\begin{IEEEproof}
	For fixed $k \in \{0,1,\dots,\D-2 \}$, consider $\p \in \pf$ with
	\begin{align}
		\p(i) &= \p(0), \forall i \in \{0,1,\dots,k \}, \\
		\p(i) &= \lambda \p(0), \forall i \in \{k+1,k+2,\dots,\D-1\}.
	\end{align}

Since $\sum_{i=-\infty}^{+\infty} \p(i) = 1$, \
\begin{align}
	2\frac{(k+1)\p(0) + (\D-k-1)\lambda \p(0)}{1-b} - \p(0) = 1,
\end{align}
and thus
\begin{align}
	\p(0) = \frac{1-b}{2(k+1) + 2(\D-k-1)\lambda -1 + b}.
\end{align}

Hence, $\p$ is specified by only one parameter $\lambda$.

The cost of $\p$ is
\begin{align}
	V(\p) &=  \sum_{i=0}^{\D-1} \p(i)w_i \\
		  &= \p(0)\sum_{i=0}^k w_i + \lambda \p(0) \sum_{k+1}^{\D-1} w_i  \\
		  &= \frac{(1-b) (\sum_{i=0}^k w_i + \lambda \sum_{i=k+1}^{\D-1}w_i  )}{2(k+1) + 2(\D-k-1)\lambda -1 + b} \\
		  &= (1-b) (C_1 + \frac{C_2}{ 2(k+1) + 2(\D-k-1)\lambda -1 + b }),
\end{align}
where $C_1$ and  $C_2$ are constant terms independent of $\lambda$. Therefore, to minimize $V(\p)$ over $\lambda \in [e^{-\e},1]$, $\lambda$ should take the extreme points, either $e^{-\e}$ or $1$, depending on whether $C_2$ is negative or positive.

When $\lambda = 1$, then the probability mass function is uniquely determined, which is $\p \in \pf$ with
\begin{align}
	\p(i) = \frac{1-b}{2\D-1+b}, \forall i \in \{0,1,\dots,\D-1\},
\end{align}
which is exactly $\p_{r}$ defined in \eqref{eqn:defprdis}  with $r = \D$.

When $\lambda = e^{-\e}$, the probability mass function is exactly $\p_r$ with $r = k+1$.

Therefore, we conclude that
\begin{align}
	V^*  = \min_{\{\mt \in \N | 1 \le \mt \le \D \}}  \sum_{i = -\infty}^{+\infty} \loss(i) \p_{\mt}(i).
\end{align}

\end{IEEEproof}

\section*{Acknowledgment}\label{sec:acknowledgment}
We thank Sachin Kadloor (UIUC) for helpful discussions, and thank Prof. Adam D. Smith (PSU) and Prof. Kamalika Chaudhuri (UCSD) for helpful comments on this work. We thank Chao Li (UMass) and Bing-Rong Lin (PSU) for pointing out the slides \cite{stairslides} to us, where  the same class of staircase mechanisms was presented under a different optimization framework. 

\bibliographystyle{IEEEtran}

\bibliography{reference}

\end{document}